\documentclass[runningheads]{llncs}

\usepackage{environ}

\newcommand{\repeattheorem}[1]{%
  \begingroup
  \renewcommand{\thetheorem}{\ref{#1}}%
  \expandafter\expandafter\expandafter\theorem
  \csname reptheorem@#1\endcsname
  \endtheorem
  \endgroup
}

\NewEnviron{reptheorem}[1]{%
  \global\expandafter\xdef\csname reptheorem@#1\endcsname{%
    \unexpanded\expandafter{\BODY}%
  }%
  \expandafter\theorem\BODY\unskip\label{#1}\endtheorem
}

\NewEnviron{repproposition}[1]{%
  \global\expandafter\xdef\csname repproposition@#1\endcsname{%
    \unexpanded\expandafter{\BODY}%
  }%
  \expandafter\proposition\BODY\unskip\label{#1}\endproposition
}

\newcommand{\repeatproposition}[1]{%
  \begingroup
  \renewcommand{\theproposition}{\ref{#1}}%
  \expandafter\proposition
  \csname repproposition@#1\endcsname
  \endproposition
  \endgroup
}
\newenvironment{claimproof}{\begin{proof}}{\vspace{-5pt}\hfill{\small $\dashv\,$}\end{proof}}

\usepackage[utf8]{inputenc}
\usepackage[hyphenbreaks]{breakurl}
\usepackage{lmodern}
\usepackage{graphicx}
\usepackage{xcolor}
\usepackage{float}
\usepackage{hyperref}
\usepackage[draft]{todonotes}
\usepackage{lineno}
\usepackage{url}
\usepackage{booktabs}       
\usepackage{nicefrac}       
\usepackage{microtype}      
\usepackage[dvipsnames]{xcolor}       
\usepackage{amsmath, amssymb, amsfonts}
\usepackage{mathtools}
\usepackage{graphicx}
\usepackage{tikz}
\usetikzlibrary{arrows, positioning}
\usepackage[normalem]{ulem}
\usepackage{enumerate}
\usepackage{xspace}
\usepackage{mathabx}
\usepackage{stmaryrd}
\usepackage{orcidlink}
\usepackage{soul}
\usepackage{dsfont}
\usepackage{enumitem}
\usepackage[capitalize,nameinlink,noabbrev]{cleveref}
\crefname{AlgoLine}{Line}{Lines}
\crefname{Algorithm}{Algorithm}{Algorithms}
\hypersetup{colorlinks=true, allcolors=blue, breaklinks=true}
\usepackage[
    backend=biber,
    sorting=nty,
    maxnames=10,
    useprefix,
    url=false,
]{biblatex}

\let\citet\textcite
\usepackage[ruled,vlined,linesnumbered]{algorithm2e}
\DontPrintSemicolon
\SetVlineSkip{1pt}
\SetCommentSty{textrm}

\DeclareSymbolFont{cmsy}{OMS}{cmsy}{m}{n}
\DeclareMathSymbol{\cmemptyset}{\mathord}{cmsy}{"3B}
\let\emptyset\cmemptyset

\newcommand{\colsub}[2]{{{#1}\ifthenelse{\equal{#2}{}}{}{, {#2}}}}

\newcommand{\tup}[1]{{\left(#1\right)}}

\newcommand{\One}{\mathbf{1}}

\newcommand{\Field}{\mathbb{F}}

\newcommand{\Nat}{\mathbb{N}}
\newcommand{\Real}{\mathbb{R}}
\newcommand{\Rat}{\mathbb{Q}}

\newcommand{\powset}[1]{{\mathcal{P}\kern-0.3em\left(#1\right)}}
\newcommand{\atp}{\operatorname{atp}}

\newcommand{\bigml}{\big\{\kern-0.4em\big\{}

\newcommand{\bigmr}{\big\}\kern-0.4em\big\}}

\newcommand{\bigmid}{\,\big|\,}

\newcommand{\vecw}{{\vec w}}
\newcommand{\vecv}{{\vec v}}
\newcommand{\vecu}{{\vec u}}
\newcommand{\vecsig}{{\vec \sigma}}
\newcommand{\vecbeta}{{\vec \beta}}
\newcommand{\vecc}{{\vec c}}
\newcommand{\vecd}{{\vec d}}

\newcommand{\chiwl}[2]{{\chi_{{#1}\text{-}\mathsf{wl}}^\tup{#2}}}
\newcommand{\chisw}[2]{{\chi_{{#1}\text{-}\mathsf{sw}}^\tup{#2}}}

\renewcommand{\hom}{\ensuremath{\mathsf{hom}}}

\newcommand{\In}{\mathchoice{\in}{\text{ in }}{\in}{\in}}
\newcommand{\Iff}{ \text{ if and only if } }

\newcommand{\An}{ \text{ and } }
\newcommand{\Or}{ \text{ or } }

\newcommand{\Where}{ \text{ where } }

\newcommand{\emptyword}{\varepsilon}

\newcommand{\canon}[1]{{\widehat{#1}}}

\newcommand{\rk}{\operatorname{rank}}

\newcommand{\sem}[1]{\llbracket {#1} \rrbracket}
\newcommand{\Ac}{\mathcal{A}}

\newcommand{\Sig}{\Sigma}

\newcommand{\RealVkVl}

\newcommand{\RatQQk}{\Rat^{Q_k\times Q_k}}

\newcommand{\OneT}{{\One\!}^\top}

\newcommand{\pin}{{+}}

\newcommand{\ip}{{i+1}}

\newcommand{\tw}[1]{\mathsf{tw}(#1)}
\newcommand{\pw}[1]{\mathsf{pw}(#1)}
\newcommand{\td}[1]{\mathsf{td}(#1)}

\newcommand{\Wl}[2]{\ensuremath{\text{${#1}$-}\mathsf{wl}^{(#2)}}}
\newcommand{\Wll}[1]{\ensuremath{\text{${#1}$-}\mathsf{wl}}}
\newcommand{\Sw}[2]{\ensuremath{\text{${#1}$-}\mathsf{sw}^{(#2)}}}
\newcommand{\Sww}[1]{\ensuremath{\text{${#1}$-}\mathsf{sw}}}

\newcommand{\Inv}{\mathcal{I}_{\mathsf{sw}}}
\newcommand{\im}{\operatorname{Im}}

\newcommand{\sig}{\sigma}

\newcommand{\hombig}[1]{\hom\big(#1\big)}

\newcommand{\hasse}[1]{\ensuremath{\mathcal{H}(#1)}}

\newcommand{\stup}[2]{\ensuremath{[\![{#1},\,{#2}]\!]}}

\newcommand{\OO}{\mathrm{O}}

\newcommand{\foc}{\ensuremath{\mathsf{C}^{k+1}}}
\newcommand{\rfoc}{\ensuremath{\mathsf{C}^{k+1}_\mathsf{R}}}

\newcommand{\bm}[1]{#1}

\def\mD{{\bm{D}}}

\def\mF{{\bm{F}}}

\def\mM{{\bm{M}}}

\def\mP{{\bm{P}}}

\def\mT{{\bm{T}}}

\newif\ifcamera
\cameratrue

\ifcamera
\author{Marek \v{C}ern\'{y}\,\orcidlink{0000-0001-6013-2054}}
\authorrunning{M. \v{C}ern\'{y}}
\institute{University of Antwerp, Antwerp, Belgium\\
\email{marek.cerny@uantwerp.be}}
\else
\author{Anonymous Authors}
\institute{Anonymous Institutes}
\hypersetup{
  pdfauthor={Anonymous authors},
  pdftitle={},
  pdfsubject={},
  pdfkeywords={}
}
\fi

\title{
    On Numbers of Simplicial Walks 
    and Equivalent Canonizations
    for Graph Recognition
}

\titlerunning{On Numbers of Simplicial Walks for Graph Recognition}

\newcommand{\rev}[2]{#2}

\begin{document}
\maketitle
\begin{abstract}
Two graphs are isomorphic exactly when they admit the same number of homomorphisms from every graph.
Hence, a graph is recognized up to isomorphism by homomorphism counts over the class of all graphs.
Restricting to a specific graph class yields some natural isomorphism relaxations and modulates recognition to particular graph properties.
A notable restriction is to the classes of bounded treewidth, yielding the isomorphism relaxation of Weisfeiler--Leman refinement (WL), as shown by Dvo\v{r}\'{a}k [JGT 2010].
The properties recognized by WL are exactly those definable in fragments of first-order logic with counting quantifiers, as shown by Cai, F\"{u}rer, and Immerman [Comb. 1992].

We characterize the restriction to the classes of bounded pathwidth by numbers of simplicial walks, and formalize it into a refinement procedure (SW).
The properties recognized by SW are exactly those definable in fragments of restricted-conjunction first-order logic with counting quantifiers, introduced by Montacute and Shah [LMCS 2024].

Unlike WL, computing SW directly is not polynomial-time in general.
We address this by representing SW in terms of multiplicity automata.
We equip these automata with an involution, simplifying the canonization to standard forward reduction and omitting the backward one.
The resulting canonical form is computable in \rev{time $O(kn^{4k})$}{time\footnote{This arXiv version improves the time comp. of the LATIN 2026 paper from 
$O(kn^{4k}).$} $O(kn^{3k})$} for any graph on $n$ vertices and the restriction to pathwidth at most $k$.

\end{abstract}


\section{Introduction}\label{sec:intro}

The societal impact of foundational research is often indirect \cite{Sto+11}, yet graph theory constitutes a remarkable exception.
Many properties of graphs apply uniformly to abstract well-defined objects (e.g., algebras or categories), as well as to concrete observed ones (e.g., proteins or molecules).
An example is the Weisfeiler--Leman algorithm (WL) \cite{Wei+68, Wei+76}, which has recently established itself as a default tool to characterize (and to design) graph properties recognized by machine learning models \cite{Xu+18, Mor+19}.

The wide applicability of WL only reflects how prominently it is situated, with connections to various research areas that were uncovered over the last decades:
the \emph{$k$-dimensional WL} ($k$-WL) recognizes the same properties as all definable formulas of 
the $(k+1)$-variable fragment of \emph{first-order logic} with counting quantifiers \cite{Cai+92} denoted by \foc, and as
the \emph{homomorphism counts} over graphs of treewidth at most $k$ (graph class $\mathcal{T}_k$) \cite{Dvo+10, Del+18}.
More closely, a~homomorphism to a graph $G$ is an edge-respecting map of vertex sets with another graph $F$ as a domain.
Indeed, WL implicitly aggregates in $G$ the homomorphic images of \emph{tree-like} graphs $F$ taken over $\mathcal{T}_k.$
Finally, \citet{Lov+67} showed that considering homomorphism counts from every graph $F$ determine $G$ up to isomorphism.

On the other hand, the nature of many graph-described objects is rather linear, \emph{path-like}, such as that of proteins folded after transcribing the ribonucleic acid sequence \cite{Dur+98}. 
Consequently, some learning methods base on sampling walks in graphs achieve the state-of-the-art results, yet being outside of WL-based characterizations \cite{Kri+22, Ton+23}.
In parallel, the arboreal and path categories \cite{Abr+23} capture rather linear objects.
This relates to the restricted-conjunction $(k+1)$-variable fragment of first-order logic with counting quantifiers (logic \rfoc), and 
properties recognizable by homomorphism counts over graphs of bounded pathwidth (graph class $\mathcal{P}_{k,1}$) \cite{Mon+24}.

This paper characterizes a recognition of these path-like graph properties.
Firstly, we describe homomorphic images over graphs in classes of bounded pathwidth (in fact, more general linear classes $\mathcal{P}_{k,h}$) by
considering the first $h$~coloring steps of $k$-dimensional WL on the graph $G$ and then determining numbers of $h$-colored sequences visited by higher-order simplicial walks in graphs.

\begin{reptheorem}{main_equivalence}\label{thm:main_equivalence}
For two graphs $G$ and $H$\!, the following conditions are equivalent:
\begin{enumerate}[label=(\roman*)]
\item $G$ and $H$ have the same numbers of $h$-colored $k$-simplicial walks.
\item $G$ and $H$ have the same homomorphism counts over graphs in $\mathcal{P}_{k,h}.$
\end{enumerate}
\end{reptheorem}

Immediately, the above theorem leads to the characterization of graph properties recognized 
by the logic of restricted-conjunction $\rfoc$ from its relation to $\mathcal{P}_{k,1}$ as shown by \citet[Corollary 5.13]{Mon+24}.
\begin{corollary}\label{cor:main_equivalence_logic} For two graphs $G$ and $H$\!, the following conditions are equivalent:
\begin{enumerate}[label=(\roman*)]
\item $G$ and $H$ have the same numbers of $1$-colored $k$-simplicial walks.
\item $G$ and $H$ have the same homomorphism counts over graphs in $\mathcal{P}_{k,1}.$
\item $G$ and $H$ satisfy the same formulas definable in logic \rfoc.
\end{enumerate}
\end{corollary}

Secondly, we formally define simplicial walk refinement (SW), analogous to WL.
Unlike WL, which assigns to every vertex tuple a single incident object (tree-like color), SW captures incidences to multiple colored simplicial walks.
This compounds into a requirement to represent exponentially many objects in the worst case.
We overcome this by representing SW indirectly via its generators, formulated as \emph{multiplicity automata} (MAs).

The remaining challenge is to compare MAs whose minimal realizations depend on the choice of basis.
We eliminate this dependence by endowing multiplicity automata with an \emph{involution} and applying forward reduction.
Canonization is particularly important in practical scenarios, where models operate on individual graphs sampled from a data distribution.

\begin{reptheorem}{main_canonization_algorithm}\label{thm:main_canonization_algorithm}
 For $h,k\ge 1$, $n\ge 0$, there exists a canonical representation $\Inv$ of size $\OO(kn^{2k})$, computable in time \rev{$\OO(k n^{4k})$}{$\OO(k n^{3k})$},
 such that for every graph $G$ on $n$ vertices, $\Inv (G, k,h)$ recognizes exactly the same graph properties as numbers of $h$-colored $k$-simplicial walks in $G.$
\end{reptheorem}

Closely related is the decision problem of \emph{homomorphism indistinguishability}: 
two graphs on $n$ vertices are homomorphism indistinguishable over $\mathcal{F}$ if they admit the same homomorphism counts over $\mathcal{F}$. 
Recently, it was shown by \citet{Sep+24} that the case $\mathcal{F} = \mathcal{P}_{k,1}$ admits a polynomial-time algorithm of complexity $\OO(k^c n^{7k+7})$ provided that $n^{k+1} \log n \ge e^{2000}$.
Our \cref{thm:main_canonization_algorithm} gives a constructive algorithm by computing $\Inv$ for both graphs and comparing them for equality.
\begin{corollary}\label{cor:main_equivalence_algorithm}
The problem of homomorphism indistinguishability over $\mathcal{P}_{k,1}$ for two graphs $G$ and $H$ on $n$ vertices 
is decidable in time \rev{$\OO(k n^{4k})$}{$\OO(k n^{3k})$}.
\end{corollary}

\paragraph{Related work.}
The equivalence problem for finite-state automata viewed as multiplicity automata and the complexity of their equivalence was originally studied by~\citet{Tze+96}. 
The complexity of $\Rat$-MA minimization was studied in~\cite{Kie+13, Kie+17} and its numerical stability in~\cite{Kie+14}. 
MAs recognize series characterized by finite-rank Hankel matrices~\cite{Fli+74}. 
When these matrices are additionally of bounded norm, a canonical form for such MAs was introduced in~\cite{Bal+15}.
Secondly, our result on homomorphism counts extends the argument of~\cite{Cer+25} for $k=1$, relating colored walks and color refinement~\cite{Bab+79}, to simplicial walks and the WL algorithm.
Used techniques of labeled quantum graphs introduced by~\citet{Lov+09} link homomorphism indistinguishability 
to WL~\cite{Dvo+10}, quantum groups~\cite{Man+20, Nig+24}, linear equations~\cite{Gro+21b}, graph spectra~\cite{Rat+23}, 
or semidefinite programming~\cite{Rob+23}.

The rest of this paper is organized as follows. 
\cref{sec:walks} introduces the central definitions of simplicial walks.
Sections~\ref{sec:auto} and~\ref{sec:homs} each establish one of our main results, Theorems~\ref{thm:main_canonization_algorithm} and~\ref{thm:main_equivalence}, respectively.  
In \cref{sec:auto}, we begin with the problem of canonical minimal realization of involution multiplicity automata over general fields and apply our findings to the representation of simplicial walk refinement.
In \cref{sec:homs}, we introduce a caterpillar decomposition formally defining the graph classes $\mathcal{P}_{k,h}$, 
and establish connections between these classes and the numbers of $h$-colored $k$-simplicial walks.
Proofs of \emph{propositions} are deferred to appendix.


\section{Preliminaries}
For $t\In \Nat$, let $[t] \coloneq\{1, 2, \dots, n\}.$
For a set $X$, let $\powset X\coloneq \{ Y \mid Y \subseteq X\}$. 
A~$t$-tuple $\vecu \In X^{t}$ has entries $u_i \coloneqq \vecu(i)$ for $i\In [t]$,
and $X^{\le t}\coloneq X^0 \cup X^{1} \cup \cdots \cup X^{t}.$
In what follows, we view $\vecu$ as a map $\vecu\colon [k] \to X$ so that $\vecu = (u_i \mid i\In [k]).$
Replacing the $p$-th element of $\vecu$ by $x$ is indicated by $\vecu[p\mapsto x]$,
and deleting the $p$-th element of $\vecu$ is indicated by $\vecu[\hat p]$ for $x \In X$ and $p \In [k].$
Finally, $\vecu v$ indicates $(u_1, u_2, \dots, u_k, v) \In X^{k+1},$ analogously $v\vecu = (v, u_1, u_2, \dots, u_k) \In X^{k+1}.$

An \emph{undirected finite graph} $G=(V, E)$ is a pair of a finite vertex set $V=V(G)$ and a symmetric edge relation $E=E(G)\subseteq V^2$ which we denote by $u v \In E$ for $u, v\In V.$
A \emph{walk} of length $t$ in a graph $G$ is a $t$-tuple of vertices $\vecu$ such that $u_i u_{i+1} \In E$ for $i \In [t-1].$
A \emph{path} is a walk in which all vertices are distinct, and a \emph{closed walk} is a walk such that $u_1 = u_t.$
For a subset of vertices $S\subseteq V$, we denote the \emph{induced subgraph} on $S$ by $G[S]$, 
that is $G[S] = (S, E \cap S^2).$

\newcommand{\notmapsto}{\mathrel{\text{\ooalign{$\mapsto$\cr\hidewidth$\mkern2mu/\mkern2mu$\hidewidth}}}}

\rev{}{\vspace{-3pt}}
\subsubsection{Atomic types and WL refinement.}
The notion of atomicity originates from atomic formulas in logic, cf. \cite{Cha+90}.
Here, an \emph{atomic type} of a graph $G$ and $k$-tuple $\vecu\In V(G)$ denoted by $\atp_k(G, \vecu)$ is 
the class of equivalence defined between $(G, \vecu)$ and $(H, \vec v)$ as $\atp_k(G, \vecu) = \atp_k(H, \vec v)$ if and only if
$u_i = u_j \;\Leftrightarrow\; v_i = v_j$ and
$u_i u_j \In E(G) \;\Leftrightarrow\; v_i v_j \In E(H)$ for all $i, j \In [k].$


For a graph $G$ and integer $k\ge 0$, a \emph{$k$-dimensional Weisfeiler--Leman refinement ($k$-WL)} first constructs a sequence of colorings for each $h\ge 0$
and $\vecu\In V(G)^k$: 
\begin{align*}
\chiwl{k}{0} &\coloneqq \atp_{k},\quad
\chiwl{k}{h+1} \coloneqq \big(\chiwl{k}{h},\, \xi^\tup h \big),\quad\text{ where }\\
    \xi^\tup h (G, \vecu) &= \bigml
     \big(\!\atp_{k+1}(G, \vecu v),
     \chiwl{k}{h}(G, \vecu[p\mapsto v]) \mid p \In [k]
     \big)
     \bigmid
        v \In V(G)
    \bigmr.
\end{align*}
Secondly, $k$-WL defines a sequence of functions
$\Wl{k}{h}(G) \coloneqq \bigml \chiwl{k}{h}(G, \vecu) \mid \vecu \In V(G)^k\bigmr,$
and finally, $\Wll{k}(G) \coloneqq \{\Wl{k}{h}(G) \mid  h \In \Nat\}.$


\rev{}{\vspace{-3pt}}
\subsubsection{Simplices and Hasse graph.}
For a given finite set of \emph{vertices} $V,$ a \emph{(finite) simplicial complex} is a family of non-empty subsets $K\subseteq \powset{V}$ satisfying $\sig\In K$ and non-empty $\tau\subseteq \sig$ implies $\tau\In K.$
Elements of $K$ are \emph{simplices}. 
A \emph{dimension} of a simplex $\sig$ is $\dim \sig = |\sig|-1$, and $\sig$ is a \emph{$d$-simplex} if $\dim \sig = d.$
The dimension of complex $K$, denoted by $\dim K$, is the supremum of dimensions of its simplices.
A \emph{$k$-skeleton} of a simplicial complex $K$ denoted by $K(k)$ is the set of all $d$-simplices $\sig \In K$ such that $d \le k.$

For a given simplicial complex $K$ on a set $V$, there is an inherited order on its simplices given by inclusion, $(K, \subseteq).$
A \emph{Hasse graph}, denoted by $\hasse K$, is a graph with $V(\hasse K) = K$ and edges given as
follows 
$\sig, \tau\In E(\hasse K)$ if and only if $\tau$ covers $\sig$ or $\sig$ covers $\tau.$ 
From the definition of a finite simplicial complex, it follows that $\hasse K$ is a finite graph. See \cref{fig:hasse_graph}~(b).

\rev{}{\vspace{-3pt}}
\subsubsection{Multiplicity automata.}
 Let $\Sigma$ be a finite set of symbols. 
We denote the set of all finite words by $\Sigma^{<\infty}$, and by $\Sigma^t\subset \Sigma^{<\infty}$ all words of length $t.$ The \emph{empty word} is denoted by $\emptyword.$
Let $\Field$ be a field. A $\Field$-\emph{multiplicity automaton} ($\Field$-MA), 
is a tuple $\Ac = (Q, \Sigma, \mM, \alpha, \eta)$, where 
$Q$ is a finite set of \emph{states},
$\Sigma$ is a finite \emph{alphabet},
$\mM\colon \Sigma \to \Field^{Q\times Q}$ is an assignment of \emph{transition matrices}, and
$\alpha^\top, \eta \In \Field^{Q}$ are \emph{an initial}  and \emph{a final} vector, respectively.
The map $\mM$ extends uniquely (as homomorphism of monoids) to all words $\vecc = c_1 c_2 \cdots c_t \In \Sigma^{<\infty}$ by defining $\mM(\vecc) \coloneq \mM(c_1) \mM(c_2) \cdots \mM(c_t).$

The \emph{word series} over $\Sigma$ with coefficients in $\Field$ is the formal series
$\Sigma^{<\infty} \to \Field.$
An $\Field$-MA $\Ac$ \emph{recognizes} word series given by $\vecc \mapsto \alpha \mM(\vecc) \eta\In \Field$ for $\vecc \In \Sigma^{<\infty}$, denoted by $\sem{\Ac}.$
Two $\Field$-MAs $\Ac_1$ and $\Ac_2$ are \emph{equivalent} if $\sem{\Ac_1} = \sem{\Ac_2}$, that is, recognize identical word series.



\section{Simplicial Walks}\label{sec:walks}
This section introduces simplicial walks, a higher-order analogue of classical walks in graphs.
We also use them to define graph decompositions in \cref{sec:homs}.

\begin{figure}[t]
    \centering
    \includegraphics[width=0.95\textwidth]{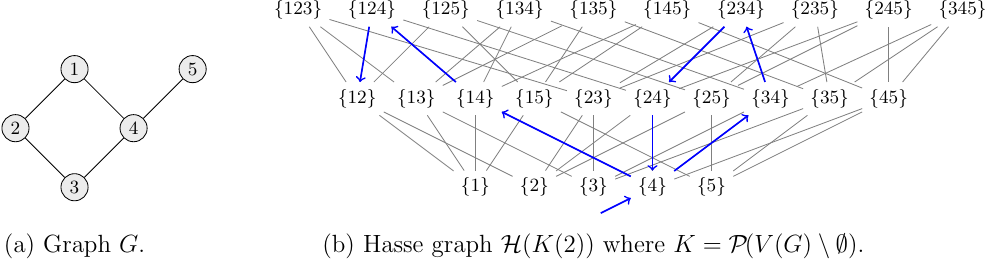}
    \caption{A $2$-simplicial walk $\vecsig$ (blue) in a graph $G.$}
    \label{fig:hasse_graph}
\end{figure}

\begin{definition}[Simplicial walk]\label{def:simplicial_complex_walk}
Let $G=(V,E)$ be a graph, $k\ge 1$ an integer, and let $K=\powset V\setminus \{\emptyset\}.$
Then a~\emph{$k$-th~simplicial walk in $G$ of length $t$} is a walk $\vecsig$ of length $t$ in the graph $\hasse {K(k)}$ such that its first simplex is of dimension $0.$
\end{definition}

It follows that for a given graph $G$, a $k$-simplicial walk $\vecsig$ of length $t$ in $G$, 
and for each $i \In [t-1]$ simplices $\sig_i$ and $\sig_\ip$ are not only adjacent in $\hasse{\powset V\setminus \{\emptyset\}}$, but moreover,
there is $u \In V(G)$ such that either $\sig_i \sqcup \{u\} = \sig_\ip\Or\sig_{i}  = \sig_\ip \sqcup \{ u \}.$ 
See \cref{fig:hasse_graph}.
We say $u$ is \emph{incoming} for edge $\sig_i\sig_\ip$ in the former case, and similarly \emph{outgoing} in the latter.
Consequently, we associate every $d$-simplex $\sig_i$ of $\vecsig$ with a tuple $\stup{\vecsig}{i} \In V(G)^{d+1}$ containing elements of $\sig_i$ ordered as their incoming edges in $\vecsig$. 
See \cref{fig:simplicial_walk}.

\begin{figure}[h]
    \centering
    \includegraphics[width=0.8\textwidth]{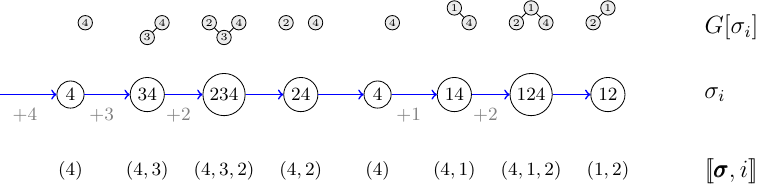}
    \caption{A simplicial walk $\vecsig$ in $G$ (Fig.~\ref{fig:hasse_graph}), with each $\sigma_i$ ordered as $\stup \vecsig i.$}
    \label{fig:simplicial_walk}
    \vspace{-10pt}
\end{figure}

\begin{definition}[Colored simplicial walk]\label{def:color_of_simplicial_walk}
Let $k, h\ge 1$, and let $\vecsig$ be a $k$-simplicial walk of length $t$ in a graph $G$,
then a \emph{$h$-color of $\vecsig$} is a word $\vecc$ such that $\vecc(i) = \chi_{k,h}(G, \stup \vecsig i)$ for $i\In [t]$,
where $\chi_{k, h}$ is given for every non-empty $\ell$-tuple $\vecu \In V(G)^{\le k+1}$ by
\begin{equation}
\chi_{k,h}(G,\vecu) \coloneqq \begin{cases}
   \chiwl{k}{k + h - \ell}(G,u_1, u_2, \ldots, u_\ell, \dots, u_\ell)
   & \quad \text{ if } \ell \le k, \\[2pt]
   \atp_{k+1}(G, u_1, u_2, \ldots, u_{k+1})
   & \quad \text{otherwise.}
\end{cases}
\end{equation}
We say that $k$-simplicial walk $\vecsig$ is an \emph{occurrence} of $h$-color $\vecc$ in $G.$
By the numbers of $h$-colored $k$-simplicial walks, we mean the numbers of such occurrences.
\end{definition}

Compared to $k$-WL refinement that colors all $k$-tuples of vertices,
we define $(k, h)$-SW refinement for all $\ell$-tuples of vertices, where $\ell\le k+1.$

\begin{definition}[SW refinement]
    Let $k,h \ge 1$ and let $G=(V,E)$ be a graph,
    then a~\emph{$(k, h)$ simplicial walk refinement}, shortly $(k, h)$-SW,
    firstly constructs a sequence of colorings for every $t\ge 0$ and $\ell$-tuple $\vecu\In V^{\le k+1}$ as follows 
    \begin{align*}
        \chisw{(k,h)}{0}(G, \vecu) &\coloneqq \bigml \chi_{k,h}(G, \vecu)\bigmr \, \text{ if }\, \ell = 1, \An \bigml \bigmr \text{ otherwise, } \\
        \chisw{(k,h)}{t+1}(G, \vecu) &\coloneqq \,\xi^\tup{t}_+(G, \vecu)\, \uplus \,\xi^\tup{t}_-(G, \vecu),
    \end{align*}
    where multisets $\xi^\tup{t}_+(G, \vecu)$ and $\xi^\tup{t}_-(G, \vecu)$ are given as 
    \begin{align*}
        \xi^\tup{t}_+(G, \vecu) =\,&\,
        \bigml \vecc \chi_{k,h}(\vecu) \,\bigmid \vecc \In \chisw{(k,h)}{t}(G, \vecu[\hat \ell\,]) \bigmr \notag\\
        &\,\text{ if } \ell > 1, \An \bigml \bigmr \text{ otherwise,}
        \notag\\
        \xi^\tup{t}_-(G, \vecu) =\,&\,\bigml \vecc \chi_{k,h}(\vecu) \,\bigmid 
        \vecc \In \chisw{(k,h)}{t}(G, \vecv),\, \vecv[\hat p] = \vecu,\,  p\in [\ell] \bigmr \\ 
        &\,\text{ if } \ell \le k, \An \,\bigml \bigmr \text{ otherwise.}
    \end{align*}
    Secondly, $(k, h)$-SW defines 
    $\Sw {(k,h)} t (G) \coloneqq \biguplus \bigml \chisw{(k,h)}{t}(G, \vecu) \mid $ $ \vecu \In V^{\le k}\bigmr,$
    and finally, $\Sww {(k,h)} (G) \coloneqq \big\{\Sw {(k,h)} t(G)\mid t\In \Nat \big\}.$
    \label{def:sw}
\end{definition}

For practicality, as clarified later in \cref{sec:homs}, the $(k,h)$-SW refinement also assigns colors to tuples with repeated vertices.
Note that every color in the image of $\chi_{k,h}$ includes its atomic type, thus records all vertex repetitions. 
This observation underlies the proof of the following proposition.

\begin{repproposition}{simpl_walks_and_walk_refiment}\label{prop:simpl_walks_and_walk_refiment}
    For $k,h \geq 1$, and a graph $G$, multiplicities in $\Sww{(k,h)}(G)$ are determined by numbers of $h$-colored $k$-simplicial walks in $G$ and vice versa.
\end{repproposition}

\section{Equivalent Automata Canonization}\label{sec:auto}
In this section, use the concept of multiplicity automata ($\Field$-MA), which can be seen as a weighted extension of finite automata. 
We define a subclass of $\Field$-MAs for which we can always find a minimal and \emph{canonical realization}.
Next, we establish automata recognizing all colored simplicial walks in a graph, and finally, obtain equivalent polynomial canonical representations.

\repeattheorem{main_canonization_algorithm}

We endow the alphabet $\Sigma$ with additional symmetry.
Generally, a tuple such as ours $(\Sigma^{<\infty},\cdot, \emptyword)$ is a monoid on the set $\Sigma^{<\infty}$ if $(\cdot)$ is an associative binary operation and $\emptyword$ is its identity element.
An \emph{involution} is a unary operation ${}^\ast\colon \Sigma^{<\infty} \to \Sigma^{<\infty}$ satisfying $(\vecc \cdot \vecd)^\ast = \vecd^\ast \cdot {\vecc}^\ast$ and $(\vecc^\ast)^\ast = \vecc$ for all $\vecc, \vecd \In \Sigma^{<\infty}.$

In the context of homomorphism tensors, \citet{Gro+21b} establish that finite-dimensional representations of monoids with involution are semisimple. 
Here, rather than relying on representation theory, we use multiplicity automata endowed with involution.
Throughout, we use $*$-notation, writing $\mM^*$ and $\vecv^*$ for adjoint matrices and vectors over $\Field$ (i.e., transposition for $\Real, \Rat$ and conjugation for $\mathbb{C}$).

\begin{definition}
A \emph{multiplicity involution automaton ($\Field$-MIA)} is a $\Field$-MA
$\Ac = (Q, \Sigma, \mM, \alpha, \eta)$ satisfying the following conditions:
\begin{enumerate}[label=(I\arabic*)]
    \item The initial and final vectors are adjoint, i.e.\ $\alpha^* = \eta.$
    \item The alphabet $\Sigma$ is endowed with an \emph{involution} ${}^*\colon\Sigma \to \Sigma,$ and for every letter $c \In \Sigma,$ it holds that $\mM(c^*) = \mM(c)^*.$
\end{enumerate}
\end{definition}
Note that by the above definition, $\mM$ extends to homomorphism of involution monoids $\Sigma^{<\infty} \to \Field^{Q\times Q}$ and $\mM(\vecc^*) = \mM(\vecc)^*$ for every word $\vecc \In \Sigma^{<\infty}.$

\begin{figure}[t]
    \centering
\begin{minipage}{0.75\textwidth}
\begin{algorithm}[H]
  \caption{Forward reduction of $\Field$-MA \cite{Tze+96, Kie+13}}\label{alg:bwbase}
        \KwIn{
        $\Ac = (Q, \Sigma, \mM, \alpha, \eta)$ 
        }
        $S \gets \emptyset$ \tcp*{finite subset of words $\Sigma^{<\infty}$}
        $\mF \gets []$ \tcp*{matrix in $\Field^{Q \times \emptyset}$}
        \renewcommand{\theAlgoLine}{3a}
        $q \gets$ [$\emptyword$] \tcp*{queue of words in $\Sigma^{<\infty}$}\label{line:queue_init_empty}
        \renewcommand{\theAlgoLine}{3b}
        $q \gets$ [$a \,|\, a \text{\textbf{ in }} (\Sigma, \le)$] \tcp*{\emph{alternative initialization} of $q$}\label{line:queue_init_alternative}
        \addtocounter{AlgoLine}{-1}
        \renewcommand{\theAlgoLine}{\arabic{AlgoLine}}
        \While{\emph{\textbf{not}} \emph{$q.$empty()}}{
            $\vecc \gets$ $q.$pop()\\
            $\gamma \gets \alpha \mM(\vecc)$\\
            \If{$\rk([\mF| \gamma]) > \rk(\mF)$}{
                $\mF \gets [\mF|\gamma]$\tcp*{add $\gamma$ as a row}
                $S \gets S \cup \{\vecc\}$\\
                    \ForEach{$a$ \emph{\textbf{in}} $(\Sigma, \le)$}{
                        $q.$push($a\vecc$)\\
                    }
            }
        }
    \Return $S, \mF$
    \tcp*{set $S$, matrix $\mF$ in $\Field^{S \times Q}$}
\end{algorithm}
\end{minipage}
\end{figure}
    
Let $\Ac$ be $\Field$-MIA on states $Q$ over alphabet $\Sigma.$
We recall  Algorithm~\ref{alg:bwbase} as in \cite{Tze+96, Kie+13} computing basis of forward space $\langle \alpha \mM(\vecw) \mid \vecw \In \Sigma^{<\infty} \rangle.$ 
Because of involution this space is isomorphic to the backward one $\langle M(\vecw)\eta \mid \vecw \In \Sigma^{<\infty} \rangle$ by $(M(\vecw)\eta)^\ast = \alpha \mM(\vecw^\ast).$
Therefore, forward reduction is sufficient to obtain minimal realization of $\sem{\Ac}.$

The reduction in Algorithm~\ref{alg:bwbase} computes base matrix $\mF \In \Field^{S \times Q},$ where $S\subseteq\Sigma^{<\infty},$ and $|S| \le |Q|.$ 
Note that the third line has two options: line~${3a}$ corresponding to the standard version, or line~${3b}$ (ours) enforcing that set $S$ does not contain $\emptyword$, which we use later.
For each letter $a\In\Sigma$ we have:
\begin{align}
    \canon \mM(a) = \mF\mM(a)\mF^\pin, \quad \canon{\alpha} = \alpha \mF^\pin, \quad \canon \eta = \mF\eta,
    \label{eq:canon}
\end{align}
where $\mF^\pin$ denotes the right inverse $\mF^*(\mF\mF^*)^{-1}.$ Equivalently, we have $\canon \mM(a)\mF = \mF\mM(a)$ and $\canon \alpha \mF = \alpha.$ 
We call $\canon \Ac = (S, \Sigma, \canon \mM, \canon \alpha, \canon \eta)$, the $\Field$-MIA constructed above, the \emph{canonical form} of $\Ac.$
By \cite{Sch+61} (see also \cite[Proposition~3.1]{Kie+20}), these automata are equivalent. Indeed, for $\vecc\In \Sigma^{<\infty}$ one has
\begin{align*}
    \sem{\Ac}(\vecc)
    = \alpha \mM(\vecc) \eta
    = \canon\alpha \mF \mM(\vecc)\eta
    = \canon\alpha \canon \mM(\vecc) \mF\eta
    = \canon{\alpha}\,\, \canon \mM(\vecc) \,\, \canon\eta
    = \sem{\canon \Ac }(\vecc).
\end{align*}

Next, we show that equivalence of $\Field$-MIAs minimized as above coincides with equality, assuming a fixed linear order on alphabet.
\begin{theorem}\label{thm:canon}
    Let $\Ac_1$ and $\Ac_2$ be two $\Field$-MIA over the alphabet $(\Sigma, \le)$, then 
    \begin{align*}\sem{\Ac_1} = \sem{\Ac_2} \quad \text{ if and only if } \quad \canon \Ac_1 = \canon \Ac_2.\end{align*}
\end{theorem}
\begin{proof}
    Let us denote $\Ac = (Q, \Sigma, \mM,\alpha, \eta)$, and $\canon\Ac = (S, \Sigma, \canon \mM,\canon \alpha, \canon \eta).$
    Now, we show that this choice is indeed canonical by showing that entries of matrices and vectors of $\canon \Ac$ are 
    determined purely by $s\coloneqq\sem{\Ac} = \sem{\canon\Ac}.$ 
    Consider $\mF\In \Field^{S\times Q}$ computed by Algorithm~\ref{alg:bwbase} with $\Ac$ on input.
    For every word $\vecc \In S$ as a coordinate of $\canon \eta$ we have 
    \[\canon\eta(\vecc) = \mF(\vecc, -)\cdot \eta = \alpha \mM(\vecc)\cdot \eta = s(\vecc).\]
    
    For other parts of $\canon\Ac$, we first prove the following claim.
    \begin{claim}[1]\label{cl:claimFF}
    Matrices $\mF\mF^*$ and $(\mF\mF^*)^{-1} \In \Field^{S\times S}$ are determined by $s.$
    \end{claim}
    \begin{claimproof}
    It suffices that for all words $\vecc, \vecd\In S$ as coordinates of $\mF\mF^*$ we have 
    \begin{align*}
    \mF\mF^*(\vecc, \vecd) 
    &= \mF(\vecc, -) \cdot \mF(\vecd, -)^*
    = \alpha \mM(\vecc) \cdot (\alpha \mM(\vecd))^* \\
    &= \alpha \mM(\vecc) \cdot \mM(\vecd)^* \alpha^*
    = \alpha \mM(\vecc) \cdot \mM(\vecd^*) \eta = s(\vecc \vecd^*).
    \end{align*}
    \end{claimproof}
    It remains to verify for each $a \In \Sigma$ and all words $\vecc, \vecd\In S$ as coordinates that 
    \begin{align*}
        \mF\mM(a)\mF^*(\vecc, \vecd) = s(\vecc a \vecd^*),\quad\An\quad \alpha \mF^*(\vecc) = s(\vecc^*).
    \end{align*}
    We obtain $\canon \mM(a)$ and $\canon \alpha$ by multiplying with $(\mF\mF^*)^{-1}.$
    Combining with Claim~(1) yields the desired statement.
\qed\end{proof}

\newcommand{\ind}[1]{\mathds{1}\!\big\{\,{#1}\,\big\}\!}

\subsection{Simplicial walk automata}
In the rest of the section, we work over the field of rationals $\Rat.$
For a given graph $G=(V,E)$ on $n$ vertices and integers $k,h\ge 1$,
we construct $\Rat$-MIA denoted by $\Ac(G, k, h)$ representing numbers of colored simplicial walks.
We first introduce building blocks for transition matrices to be assigned to the letters of our alphabet with involution.
Let the subset $Q_k\subseteq V^{\le k}$ contain all non-empty tuples \emph{without repeating vertices}.
We use the notation $\ind{\phi}$ indicating by value $1\In\Rat$ if the condition $\phi$ is true, and $0\In\Rat$ otherwise.

\begin{figure}[t]
    \centering
    \includegraphics[width=0.95\textwidth]{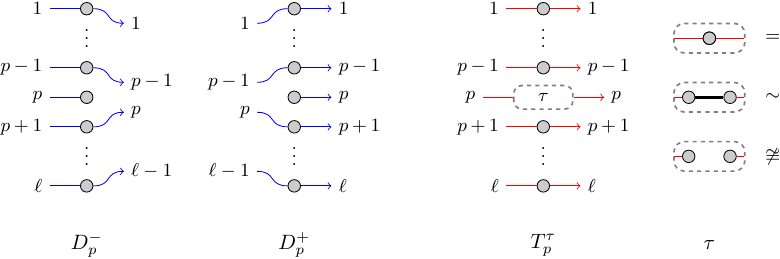}
    \caption{
         Diagrammatic view of transition matrices between $\ell$-tuple states in the simplicial walk automaton.
    }
    \label{fig:automata_base}
\end{figure}

\subsubsection{Transitions.}
We define matrices to construct the $\Rat$-MIA, see also \cref{fig:automata_base}.
First, the \emph{outgoing matrix} $\mD^-_p$ and \emph{incoming matrix} $\mD^+_p$ both in $\RatQQk$
 are defined for each $p\In[k]$ 
 and all tuples $\vecu, \vecv \In Q_k$ as 
\begin{align*}
    \mD^-_p(\vecu, \vecv) &\coloneqq \ind{\vecu[\hat p] = \vecv}, \quad\An\quad
    \mD^+_p(\vecu, \vecv) \coloneqq \ind{\vecu = \vecv[\hat p]}.
\end{align*}
Next, for an atomic type $\tau$, we define the \emph{atomic-type matrix} $\mT^\tau_p\In \RatQQk$ 
is defined for each $p\In[k]$ 
 and all tuples $\vecu, \vecv \In Q_k$ as 
\begin{align*}
    \mT^\tau_p(\vecu, \vecv) &\coloneqq \ind{\atp_2(G, \vecu(p), \vecv(p)) = \tau \An \vecu[\hat p] = \vecv[\hat p]}.
\end{align*}

Finally, for a given coloring $\chi(G, -)\colon V\to C$ and its color $a \In C$,
we define the \emph{color-partition} matrix $\mP_a\In \RatQQk$ as
\begin{align*}
    \mP_a(\vecu, \vecv) &\coloneqq \ind{\chi(G, \vecu) = a \An \vecu = \vecv}.
\end{align*}

\begin{proposition}\label{prop:matrices}
    For all $p\In [k]$, atomic types $\tau$ and colors $a, b\In C$ of a coloring $\chi(G, -)\colon V\to C$ the following properties hold.
    \begin{enumerate}[label=(\roman*)]
        \item Matrix $\mP_a$ is a projection, $\mP_a^2 = \mP_a$, and $\mP_a \mP_b = 0$ if $a\ne b.$
        \item Matrices $\mT^\tau_p$ and $\mP_a$ are symmetric, and it holds $\left(\mD^+_p\right)^\top = \mD^-_p.$
    \end{enumerate}
\end{proposition}
\begin{proof}
    From the definition.
\qed\end{proof}

\subsubsection{Alphabet.}
The alphabet $\Sig = \Sig(G, k, h)$ contains symbols of the form $(a, s, b)$ or $(a, \tau, p,  b)$, where $a, b$ are colors in $C$ of coloring $\chi(G, -)\colon V \to C$, and sign $s$ in $\{-, +\}$; or atomic type $\tau \In Y$ of $\atp_2(G, -)\colon V \to Y$ and index $p$ in $[k].$ The involution $\Sig \to \Sig$ is given by
\begin{align*}
    (a, s, b)^* = (b, -s, a)\quad\An\quad 
    (a, \tau, p, b)^* = (b, \tau, p, a).
\end{align*}
For letters $(a, s, b), (a, \tau, p, b) \In \Sig$, we define the following transition matrices
\begin{align*}
    \mM\!\left((a, s, b)\right) = \sum_{p\in[k]} \mP_a \mD^s_p \mP_b \quad\An\quad
    \mM\!\left((a, \tau, p, b)\right) = \mP_a \mT^\tau_p \mP_b.
\end{align*}
We remark $\mM((a, s, b)^*) = \mM((a, s, b))^\top$ and $\mM((a, \tau, p, b)^*) = \mM((a, \tau, p, b))^\top,$ by (ii) of \cref{prop:matrices}.

\begin{definition}[Simplicial walk automaton]
    Let $k,h\ge 1$ and $G=(V,E)$ be a~graph, then a \emph{simplicial walk automaton} $\Ac(G, k, h)$ is $\Rat$-MIA $(Q_k, \Sig, \mM, \OneT, \One)$,
    where $Q_k$ is the set of states, $\Sig = \Sig(G, k, h)$ is the alphabet, the transition matrices are $\mM((a, s, b))$ and $\mM((a, \tau, p, b))$ in $\RatQQk$ for $(a, s, b), (a, \tau, p, b) \In \Sig$, and the initial and the final vector are all-one vectors in $\Rat^{Q_k}.$
    \label{def:sw_auto}
\end{definition}

\begin{figure}[t]
    \centering
    \includegraphics[width=0.7\textwidth]{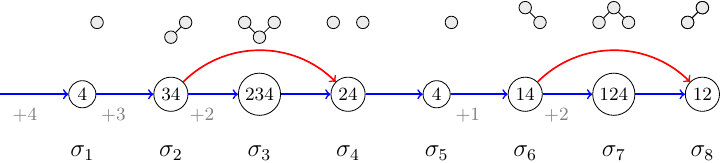}
    \caption{Illustration of $2$-simplicial walk (blue) in $G$ represented by the set of states $Q_2$ that excludes $3$-tuples, using atomic-type matrices (red).}
    \label{fig:automata_base}
\end{figure}

\begin{repproposition}{swa_equivalence}\label{prop:swa_equivalence}
    Let $k,h\ge 1$ and $G$ be a graph, then rational series $\sem{\Ac(G, k, h)}$ determines $\Sww {(k,h)} (G)$ and vice versa.
\end{repproposition}
\begin{proof}[sketch]
    The key idea is to avoid using $(k+1)$-tuples explicitly as states, which would significantly increase complexity. 
    Instead, we use atomic type matrices: $(k+1)$-tuples are colored only by $\atp_{k+1}$, which decomposes into the $\atp_k$ of neighboring tuples and the $\atp_2$ of the vertices where they differ. See \cref{fig:automata_base}. Full proof is in \cref{app:auto_proofs}.
\end{proof}

\begin{theorem}\label{thm:complexity}
    Let $k,h\ge 1$, and let $G$ be a graph, then there is an implementation of Algorithm~\ref{alg:bwbase} computing outputs for the simplicial walk automaton $\Ac(G, k, h)$ in time \rev{$\OO(k n^{4k}).$}{$\OO(k n^{3k}).$}
\end{theorem}
\begin{proof}
    Let $\Sigma = \Sigma(G, k, h)$, and $n = |V(G)|.$
    Following \citet{Tze+96} for general $\Rat$-MA, the algorithm extends matrix $F$ (line~${8}$) at most $|Q_k|$-times
    enqueuing at most $|Q_k||\Sigma|$ elements (line~${11}$).
    Each element of queue undergoes the matrix-rank query (line~${8}$) taking naively $\OO(|Q_k|^3)$, e.g., using Gaussian elimination.
    By keeping and extending a copy of matrix $F$ in reduced echelon form throughout the process, \citet{Cor+07} reduce query-time to $\OO(|Q_k|^2)$.
    This results in a time complexity of $\OO(|\Sigma||Q_k|^4)$ bounded by $\OO(k n^{5k})$, since $|\Sigma| \le 5 k n^{2k}$ and $|Q_k| \le n^{k}.$

    For our SW automata, to obtain $\chi(G, -):V\to C$, we first run $h'=\min(k+h, n^k)$ iterations of $k$-WL in time $\OO(kn^{2k+1}).$ 
    Next, observe that any word containing two consecutive letters $(a, \dots, b),(c, \dots, d)\In \Sigma$
    is of multiplicity $0$ by (i) of \cref{prop:matrices} if $b\ne c.$
    To that end, if $\vecw \neq \emptyword$, we only need to enqueue $5 k n^k$ elements (line~${11}$),
    resulting in a time complexity of $\OO(k n^{4k}).$

    \rev{}{
    Finally, to compute $\gamma= \alpha M(\vecw)= \alpha M(\vecw' (c,\ldots,d)) = \alpha M(\vecw') T P_d$ (line~${6})$,
    we not only cache $\alpha M(\vecw'),$ but also $\alpha M(\vecw') T.$ 
    Multiplication by $P_d$ then sets coordinates outside of $d$ to zero.
    Hence, it suffices to check that $\rk\left(\begin{psmallmatrix}F\\\gamma\end{psmallmatrix}P_d\right) > \rk \left(FP_d\right)$ 
    separately for each $d$-induced column subset (of size $q_d$) of the partition on $Q_k$.
    This bounds query-times to $\sum_{d\in C} q_d^2 \le n^{2k}$ for each $T$ ($\le 5k$ choices) and word length $|\vecw|\le n^k,$ resulting in the desired time complexity $\OO(kn^{3k}).$}
\qed\end{proof}

\subsection{Invariant of simplicial walk automaton}
\rev{}{\vspace{-0.5em}}
For a graph $G$, the simplicial walk automata $\Ac = \Ac(G, k, h)$ is $\Rat$-MIA, and therefore, there exists its canonical form by \cref{thm:canon}.
More closely, we use Algorithm~\ref{alg:bwbase} in the variant of line~${3b}$ ensuring $S$ does not contain $\emptyword.$
Indeed, for $\Ac$ we have dependence $\alpha = \alpha M(\emptyword) = \sum_a \alpha M((a, =, 1, a)).$ 
Hence, the forward space remains spanned by rows of matrix $\mF$, obtaining $\canon \Ac$ in \cref{eq:canon}.

We define a \emph{canonical invariant} of the simplicial walk automaton $\Ac$ using automaton $\canon{\Ac}$ as follows
\begin{align}\label{eq:inv}
    \Inv (G, k, h) &= 
    \left(
        \canon{\alpha},\,
        \canon{\eta},\,
        \canon{\mD}^+\!,\,
        \canon{\mD}^-\!,\,
        \canon{\mT}^\tau_p
        \bigmid p \In [k],\, \tau \In Y
    \right),\quad \Where\\[4pt]
    \canon{\mD}^s &= \sum_{a,b} \canon{\mM}((a, s,b)),\quad\An \quad\canon{\mT}^\tau_p = \sum_{a,b} \canon{\mM}((a, \tau, p, b)) \notag
\end{align}
for sign $s\In \{+, -\}$, atomic type $\tau \In Y$ and $a,b$ ranging over $C$ of coloring $\chi_{h,k}(G, -)\colon V(G) \to C.$
Finally, we show that invariant of \cref{eq:inv} provides desired representation.

\paragraph{Proof of \cref{thm:main_canonization_algorithm}.}
For any graph $G$ on $n$ vertices, construct its simplicial walk automaton $\Ac.$
Since $\Ac$ is a $\Rat$-MIA, apply \cref{thm:canon} to obtain the minimized automaton $\canon{\Ac}.$
This is computable in time \rev{$\OO(k n^{4k})$}{$\OO(k n^{3k})$} by \cref{thm:complexity}, and the size of $\Inv(G, k, h)$ is $\OO(k n^{2k})$ as defined.

Consider another graph $H.$ 
Then by \cref{prop:swa_equivalence}, $\canon{\Ac}(G, k, h) = \canon{\Ac}(H, k, h)$ if and only if $G$ and $H$ agree on their numbers of $h$-colored $k$-simplicial walks.
By definition of $\Inv$, we have that $\canon{\Ac}(G, k, h) = \canon{\Ac}(H, k, h)$ implies $\Inv(G, k, h) = \Inv(H, k, h).$
For the opposite direction it suffices to reconstruct $\canon{\Ac}$ from $\Inv.$

Each transition matrix $\mD$ of $\Ac$ is adjoined by a pair of partition matrices $\mP_x$ and $\mP_y$ for some colors $x, y$, i.e.  $\mD = \mD_{x,y}  =\mP_x\mD'\mP_y\In \Rat^{Q_k\times Q_k}.$ 
It follows that entry $\mD(\vecu, \vecv)$ is zero if $\chi_{k,h}(G, \vecu) \ne x$ or $\chi_{k,h}(G, \vecv) \ne y.$
Similarly, we consider its minimized version $\canon{\mD}\In\Rat^{S\times S}$ (such that $\emptyword$ is not in $S$) with entries 
\begin{align}
    \canon{\mD}(\vecc a, \vecd b) &=
    (\mF P_x \mD' P_y \mF^\pin) (\vecc a, \vecd b) = 
    (\mF P_a P_x \mD' P_y P_b \mF^\pin) (\vecc a, \vecd b), \label{eq:canon_matrix_entries}
\end{align}
that are equal to zero if $a\neq x$ or $b\neq y$ by (i) of~\cref{prop:matrices}.
Therefore, desired $\canon{\mD}_{x,y}$ is obtained from $\sum_{a,b} \canon{\mD}_{a,b}$ by all selecting entries of words in $S$ ending with $x$, resp. $y$, for the first coordinate, resp. for the second coordinate. 
Other entries are left out zero due to \cref{eq:canon_matrix_entries}.

\rev{}{\vspace{-0.5em}}
\section{Graph Recognition by Homomorphism Counts}\label{sec:homs}
\rev{}{\vspace{-0.5em}}
In this section, we formally define the graph class $\mathcal{P}_{k,h}$ (\autoref{def:cd}),
\rev{and show the following correspondence between homomorphism counts over $\mathcal{P}_{k,h}$ and numbers of $h$-colored $k$-simplicial walks.}{
and show the following connection to numbers of $h$-colored $k$-simplicial walks.}

\repeattheorem{main_equivalence}

Given two undirected graphs $F=(V(F), E(F))$ and $G=(V(G), E(G))$, 
a mapping between vertices $\varphi:V(F) \to V(G)$ is a~\emph{graph homomorphism} $F\to G$
if it respects edges, that is, there exists $\varphi(u)\varphi(v)\In E(G)$ for each $uv\In E(F).$
The homomorphism count $F$ to $G$ is denoted by $\hom(F, G) = \left|\{\varphi \mid \varphi\colon F\to G\}\right|.$

For simplicity of notation, we state the following definition equivalent to existence of path decomposition of width at most $k,$ as detailed in \cref{app:decompositions}.
\begin{definition}
    Let $F$ be a graph, and let $\vec\beta$ be a $k$-simplicial walk in $F$. 
    We say that $\vec\beta$ is \emph{decomposing} in $F$ if
    \begin{enumerate}[label=(\roman*)]
        \item[(D1)] Every $u \In V(F)$ is incoming exactly once in $\vecbeta$.
        \item[(D2)] Every edge $uv \In E(F)$ is contained in at least one $\beta_i$.
    \end{enumerate}
\end{definition}
\begin{repproposition}{decomposing_simplicial_walk}\label{prop:decomposing_simplicial_walk}
    Every graph has a path decomposition of width at most $k$ if and only if it has a $k$-simplicial walk that is decomposing.
\end{repproposition}

\subsection{Labeled and quantum graphs}
We follow the algebraic approach by Lov\'{a}sz and Szegedy~\cite{Lov+09}. 
A~\emph{$k$-labeled graph}, 
denoted as $F^\vecu$, is a graph $F$ with a distinguished $k$-tuple of \emph{labels} $\vecu \In V(F)^k.$
Given two labeled graphs $F^\vecu$ and $G^\vecv$,
a \emph{(labeled, graph) homomorphism} $\varphi\colon F^\vecu \to G^\vecv$
that is homomorphism of graphs $F \to G$ that additionally respects labels, that is,
$\varphi(u_i) = v_i$ for each $i \In [k].$
The homomorphism count $F^\vecu$ to $G^\vecv$ is denoted by $\hom(F^\vecu, G^\vecv).$

A \emph{quantum $k$-labeled graph} is a formal finite \emph{$\Real$-linear combination} of $k$-labeled graphs.
For a quantum $k$-labeled graph $F^\vecu = \sum_{i=1}^m c_i F_i^{\vecu_i}$, 
every coefficient $c_i$ is in $\Real$ and $F_i^{\vecu_i}$ is $k$-labeled graph we call a \emph{constituent of} $F^\vecu.$
Homomorphism counts extend linearly from constituents to $F^\vecu$ as follows
\begin{align*}
    \hombig{\sum_{i=1}^m c_i F_i^{\vecu_i}, G^\vecv} = \sum_{i=1}^m c_i \cdot \hombig{F_i^{\vecu_i}, G^\vecv}.
\end{align*}

For (standard) $k$-labeled graphs, $G^\vecu, H^\vecv$, we define their product as 
$G^\vecu\cdot H^\vecv$ as their disjoint union 
$G \sqcup H$ where we additionally identified vertices $u_i \In V(G)$ and $v_i \In V(H)$ for each $i \In [k].$
Note that the product can introduce self-looping edges and multiedges.
If multiedge occurs, one can remove this duplication without a change of homomorphism count.
We have the following identity:
\begin{align*}
    \hombig{K^\vecu \cdot F^\vecv,G^\vecw} = \hombig{K^\vecu,G^\vecw} \cdot \hombig{F^\vecv,G^\vecw},
\end{align*}
for any $k$-labeled graph $K^\vecu$ and any $k$-labeled graphs $F^\vecv, G^\vecw.$ 
This identity also naturally extends to quantum $k$-labeled graphs.

For the definition pebble forest covers, we refer to \cite{Abr+17, Rat+23} or \cref{app:pebble_forest_covers}. 
Denote the class of all $\ell$-labeled graphs such that $\ell \le k$ with a \emph{$k$-pebble forest cover of depth $d$} by $\mathcal{LT}_{k}^d.$
We next recall an important characterization of $\mathcal{LT}_{k}^d.$

\begin{theorem}[\cite{Dvo+10, Abr+17, Rat+23}]\label{thm:wl_homomorphism_counts}
Let $k\ge 1$, $d \ge 0$ be integers, and $G^\vecu$ and $H^\vecv$ be two $k$-labeled graphs, then 
\[\chiwl{k}{d}(G, \vec u) = \chiwl{k}{d}(H, \vec v) 
\quad\text{if and only if}\quad
\hom(\mathcal{LT}^d_k, G^\vecu) = \hom(\mathcal{LT}^d_k, H^\vecv).\]
\end{theorem}

We use a known result that is stronger than the forward \rev{implication}{part} of \cref{thm:wl_homomorphism_counts}.
\begin{lemma}\label{lem:wl_to_quantum_graph}
Let $k\ge 1$, $d \ge 0$, $n \ge 1$ be integers, 
and let $c$ be a color in $\im(\chiwl{k}{d}).$
Then there exists a quantum $k$-labeled graph $F^\vecu$ with constituents in $\mathcal{LT}^d_k$
such that for every $k$-labeled graph $G^\vecv$ on $n$ vertices it holds 
\begin{equation}
    \hom(F^\vecu, G^\vecv) \text{ equals } 1 \text{ if } c = \chiwl{k}{d}(G, \vecv), \An 0  \text{ otherwise. }
    \label{eq:quantum_labeled_model}
\end{equation}
\end{lemma}
\begin{proof}
We fix $n$ and a color $c$ in the image of $\chiwl{k}{d}.$ 
Then by {\cite[Claim~2, Theorem~5.5.3]{Gro+17}} (see also \cite{Imm+90, Cai+92}), 
there exists a formula $\varphi_c(x_1, \dots, x_k)$  with at most $d$ quantifiers in $\foc$ such that for all graphs $G$ on at most $n$ vertices 
\begin{align*}
 c = \chiwl{k}{d}(G, \vecv) 
\quad\text{ if and only if }\quad
(G, \vecv) \models \varphi_c(x_1, \dots, x_k).
\end{align*}
Finally, by {\cite[Lemma~6]{Dvo+10}, \cite[Lemma A.5]{Rat+23}}, there exists a quantum $k$-labeled graph $F^\vecu$ with constituents in $\mathcal{LT}^d_k$ such that 
for all $k$-labeled graphs $G^\vecv$ on $n$ vertices holds
if $(G, \vecv) \models \varphi_c(x_1, \dots, x_k)$ then $\hom(F^\vecu, G^\vecv) = 1$, and otherwise $\hom(F^\vecu, G^\vecv) = 0.$
\qed\end{proof}

\subsection{Caterpillar decomposition and simplicial walks}
In this final part, we introduce caterpillar decomposition, situate it within the context of known results, and prove \cref{thm:main_equivalence}.
\begin{definition}\label{def:cd}
    Let $F$ be a graph. A \emph{caterpillar decomposition} (CD) of $F$ 
    is a sequence $\vec D  =\left((\beta_i, L_i) \mid i \In [t]\right)$, where $L_i$ is an~induced subgraph of~$F$ and $\beta_i \subseteq V(L_i)$  for each $i \In [t]$ such that
    \begin{enumerate}
        \item[(C1)] Every vertex $u\In V(L_i)\setminus \beta_i$ has only neighbors in $L_i.$
        \item[(C2)] Simplicial walk $\vecbeta$ is \emph{decomposing} in $S = F\left[\beta_1 \cup \beta_2 \cup \dots \cup \beta_t\right].$
        \item[(C3)] Every labeled graph $L_i^{\stup \vecbeta i}$ is in $\mathcal{LT}_{k}^{d_i}$ for some $d_i$ or in $\mathcal{LT}_{k+1}^{0}$ and $d_i = 0.$
    \end{enumerate}
    We call the induced subgraph $S$ the \emph{spine} and the subgraphs $L_i$ the \emph{legs} of $\vec D.$
    A \emph{width} of $\vec D$ is the minimal $k$ such that $\vecbeta$ is a $k$-simplicial walk, and $t$ is a \emph{length} of  $\vec D$. 
    A \emph{height} of $\vec D$ is $\max_{i \In [t]} (|\beta_i| + d_i - k)$.

    We denote the graph class having CD of width at most $k$ and height at most~$h$ by $\mathcal{P}_{k,h},$
    and by $\mathcal{P}_{k,h}^t\subseteq \mathcal{P}_{k,h}$ its subclass having such CD of length $t+1.$
\end{definition}

\begin{repproposition}{cd_path_decomposition}\label{prop:cd_path_decomposition}
    For $k\ge 1$, $\mathcal{P}_{k,1}$ is the class of graphs of pathwidth at most~$k,$
    and for $h\ge 1$, the graph class $\mathcal{P}_{k,h}$ is a subclass of $\mathcal{P}_{k+h-1,1}.$
\end{repproposition}

It follows from \cref{lem:wl_to_quantum_graph} that the number of homomorphisms from any $k$-labeled $L^\vecu \In \mathcal{LT}_k^d$
to a $k$-labeled graph $G^\vecv$ is determined exactly by the color $c = \chiwl{k}{d}(G,\vecv).$
The following definition applies this fact.
\begin{definition}\label{def:legformula}
    Let $F$ be a graph and $\vec D = ((\beta_i, L_i) \mid i \In [t])$ be its caterpillar decomposition,
    let $G$ be a graph and $\vecc$ be colored $k$-simplicial walk,
    such that for each $i \In [t]$ it holds
    $\vecc(i) \In \im(\chi_{k,h}(G, -))\text{ implies }L_i^{\stup \vecbeta i} \In \mathcal{LT}_k^{d_i}\cup \mathcal{LT}_{k+1}^{0}.$
    Then we write 
    $\hom\big((F, \vec D), (G, \vecc)\big) \coloneq \prod_{i=1}^{t} \hombig{L_i^{\stup \vecbeta i}\!,\, \vecc(i)}.$
\end{definition}

\begin{lemma}\label{lem:homformula}
    Let $F\In \mathcal{P}_{k, h}^t$ be a graph with a caterpillar decomposition $\vec D$ of width at most $k$ and length at most $t$,
    and let $G$ be any graph. Then it holds
    \begin{align*}
        \hom(F, G) = 
        \sum_{\vecc}\mu(G, \vecc) \cdot
        \hombig{(F, \vec D), (G, \vecc)},
    \end{align*}
    where the sum ranges over elements of $\Sw {(k,h)} t (G)$, and $\mu(G, \vecc)$ denotes the multiplicity of $\vecc$ in $\Sw {(k,h)} t (G).$
\end{lemma}
\begin{proof}
Let $\vec D = ((\beta_i, L_i) \mid i \In [t])$ with spine $S$.
Consider $\psi\colon V(S) \to V(G)$ and define $\vecc$ by $\vecc(i) = \chi_{k,h}(G, \psi(\stup{\vecbeta}{i}))$ for each $i \In [t].$
Apply $\psi$ element-wise to a $\ell$-tuple of vertices, that is $\psi((u_1, \ldots, u_\ell)) = (\psi(u_1), \ldots, \psi(u_\ell)).$
In graph $G$, the multiplicity $\mu(G, \vecc)$ counts exactly the number of images of $\psi$ in $G.$
Since the legs of CD are disjoint outside the spine that is fixed by $\psi$, the number of homomorphisms from $F$ to $G$ extending $\psi$ is given by the product of homomorphism counts from each leg independently, as given in \cref{def:legformula}.
\qed\end{proof}

\begin{lemma}[{\cite[Lemma~B.1]{Cer+25}}]\label{lem:catgnn_faithful_product}
    For every integer $t\ge 0$ and natural numbers
    $D_0, D_1, D_2, \dots, D_t$
    there are integral exponents $s_1, s_2$ $\dots$, $s_t$
    such that $D_0 \le 2^{s_1}$,
    and such that
    all $t$-tuples of natural numbers $(d_1, d_2, \dots, d_t)$
    that satisfy $1 \le d_i < D_i$ for $i \In [t]$,
    are injectively represented by the product
    $p = d_1^{s_1} \cdot d_2^{s_2} \cdots \cdot d_t^{s_t}.$
\end{lemma}

\begin{lemma}[{\cite[Lemma~B.2]{Cer+25}}]\label{lem:catgnn_distinct_exponents}
    Let $\vec x$ be a $m$-tuple in $\Nat^m_{\ge 1}$ with mutually distinct elements,
    that is $\vec x_i \neq\vec x_j$ whenever $i\neq j$, 
    and let $m$-tuples  $\vec a$, $\vec b$ in $\Nat^m$ be such that 
    $\vec a \neq \vec b$, that is $a_i \neq b_i$ for existing $i\In[m].$
    Then there is $k$ in $\Nat$ such that
    \begin{align}
        \sum_{i=1}^{m} a_i x_i^k \neq \sum_{i=1}^{m} b_i x_i^k.
        \label{eq:distinct}
    \end{align}
\end{lemma}

\begin{theorem}\label{thm:walks}
    Let $k, h \ge 1, t \ge 0$ be integers, and $G$ and $H$ be graphs, then
    $\Sw {(k,h)} t (G) = \Sw {(k,h)} t (H) \Iff\hom(\mathcal{P}_{k,h}^t, G) = \hom(\mathcal{P}_{k,h}^t, H).$
\end{theorem}
\begin{proof}
    Assume $\Sw {(k,h)} t (G) = \Sw {(k,h)} t (H).$
    By \cref{lem:homformula}, for every $F\In\mathcal{P}_{k,h}^t$ 
    the number $\hom(F, G)$ is determined by the multiplicities in $\Sw {(k,h)} t (G)$.
    Analogously for $H$, we obtain $\hom(\mathcal{P}_{k,h}^t, G) = \hom(\mathcal{P}_{k,h}^t, H).$

    For the backward direction, we assume $\Sw {(k,h)} t (G) \neq \Sw {(k,h)} t (H),$
    and \emph{aim} to find a graph $F\In\mathcal{P}_{k,h}^t$ for which the homomorphism counts differ.
    
    By \cref{lem:wl_to_quantum_graph}, every $h$-color $e$ of $\ell$-tuple in $\im(\chi_{k,h}(G, -))$ is modeled (\cref{eq:quantum_labeled_model}) by a quantum $\ell$-labeled graph $K_e^\vecu$.
    Fix order on $\im(\chi_{k,h}(G, -))=\{e_1, e_2, \ldots, e_C\}$ and let $L_{e_i}^\vecu = (i+1) \cdot K_{e_i}^\vecu$ for each $i \In [C].$
    
    For $n = |V(G)|$ and any $\vecc \In \Sw {(k,h)} t (G)$, bound multiplicity of $\mu(G, \vecc)$ by $D_0 = n^{kt}$, and set $D_1 = D_2 = \ldots = D_t = C+2.$
    By \cref{lem:catgnn_faithful_product}, there are integral exponents $s_1, s_2, \ldots, s_t$ such that $D_0 \le 2^{s_1}$ and for each $d_i < D_i$ there is an injective function given by $(\mu(G, \vecc),\vecc(1), \dots, \vecc(t)) \mapsto$
    \begin{align*}
        \mu(G, \vecc)
        \prod_{i=1}^{t} 
        \left(\hombig{L_{\vecc(i)}^{\stup \vecbeta i}, \vecc(i)}\right)^{s_i}
        =
        \mu(G, \vecc)
        \prod_{i=1}^{t} 
        \hombig{\left(L_{\vecc(i)}^{\stup \vecbeta i}\right)^{s_i}\!,\, \vecc(i)}.
    \end{align*}
    Finally, apply \cref{lem:catgnn_distinct_exponents} by setting $\vec x, \vec a, \vec b \In [C+2]^t$ as 
    \begin{align*}
        \vec x(j) &= \prod_{i=1}^{t} 
        \left(\hombig{L_{\vecc_j(i)}^{\stup \vecbeta i}, \vecc_j(i)}\right)^{s_i},\quad
        \vec a(j) = \mu(G, \vecc_j),\quad
        \vec b(j) = \mu(H, \vecc_j),
    \end{align*}
    where index $j$ ranges over words $\vecc_j \In \Sw {(k,h)} t (G).$
    As a result, we found a sufficiently large $k$, 
    such that quantum graph $\bar{F}$ obtained from CD $\bar{\vec D} = ((\bar{\beta}_i, \bar{L}_i) \mid i \In [t])$, where ${\bar{L}}_{i}^{\stup {\bar{\vecbeta}} i } = \left(L_i^{\stup \vecbeta i}\right)^{k \cdot s_i},$ 
    satisfies $\hom(\bar{F}, G) \neq \hom(\bar{F}, H).$
    Hence, there necessarily exists a constituent $F$ of $\bar{F}$, such that $\hom(F, G) \neq \hom(F, H).$
    By definition of $\bar{\vec D}$, graph $F$ belongs to $\mathcal{P}_{k,h}^t.$ 
\qed\end{proof}

\paragraph{Proof of \cref{thm:main_equivalence}.} Follows from \cref{thm:walks} and \cref{prop:simpl_walks_and_walk_refiment}.

\section{Conclusion}
We introduced simplicial walks as a transparent higher-order analogue of walks. 
These walks, once colored, characterize graph properties recognized by homomorphism counts from classes such as $\mathcal{P}_{k,h}$. 
Building on this, we developed the SW refinement and SW automata, yielding more general and asymptotically faster algorithms for homomorphism-based graph property recognition.
For a subsequent notion of homomorphism indistinguishability and $h=1$ this improves asymptotically upon \cite{Sep+24}.
A key technical contribution is the use of multiplicity involution automata, where forward reduction not only minimizes but also eliminates dependence on the choice of basis.
Moreover, our canonical representations of SW automata are asymptotically smaller than those of multiplicity automata.

Our main theorem establishes a correspondence between numbers of $h$-colored $k$-simplicial walks, homomorphism counts over $\mathcal{P}_{k,h}$, and formulas definable in the restricted-conjunction logic $\rfoc$. 
This extends the classical correspondence linking Weisfeiler--Leman refinement, homomorphisms over $\mathcal{T}_k$, and definability in logic $\foc$. 
The combination of WL colors and simplicial walks, reflected in caterpillar decompositions and in the transition matrices of SW automata, offers a promising framework: it characterizes homomorphism images of path-like structures in generally shorter sequences than homomorphism tensors \cite{Gro+21b}.

A subsequent open question is whether the number of steps of $(k,h)$-SW refinement before it stabilizes can be bounded better than $\OO(n^{k})$.
Known results in \cite{Lic+19, Gro+23} inspire our question by upper-bounding this number for WL refinement.
Answering this would lead towards uncovering the nature of relations between restrictions over which we count homomorphisms, and fixed points of possible induced refinements.

\paragraph{Funding.}
This work was supported by the University of Antwerp (BOF, Doctoral Project 47103).

\paragraph{Acknowledgements.}
The author is grateful to Peter Zeman for guidance on the structure of the manuscript, and to Floris Geerts and Guillermo A. Pérez for helpful discussions on framing this work.

\printbibliography[heading=bibintoc]

@article{Kie+20,
  author       = {Stefan Kiefer},
  title        = {Notes on Equivalence and Minimization of Multiplicity Automata},
  journal      = {CoRR},
  volume       = {abs/2009.01217},
  year         = {2020},
  url          = {https://arxiv.org/abs/2009.01217},
  eprinttype    = {arXiv},
  eprint       = {2009.01217},
  bibsource    = {dblp computer science bibliography, https://dblp.org}
}

@article{Sch+61, 
  title={ On the definition of a family of automata},
  volume={4}, 
  DOI={10.2307/2271115},
  journal={Information and control}, 
  author={M. P. Schützenberger.}, 
  year={1961}, 
  pages={pp.245–270.}
}

@article{Lov+09,
author = {Lovász, László and Szegedy, Balázs},
title = {Contractors and connectors of graph algebras},
journal = {Journal of Graph Theory},
volume = {60},
number = {1},
pages = {11-30},
doi = {10.1002/jgt.20343},
url = {https://onlinelibrary.wiley.com/doi/abs/10.1002/jgt.20343},
year = {2009}
}

@article{Lov+67,
  author = {Lov{\'a}sz, L.},
  title = {Operations with structures},
  journal = {Acta Mathematica Academiae Scientiarum Hungarica},
  volume = {18},
  number = {3},
  pages = {321--328},
  year = {1967},
  month = {sep},
  issn = {1588-2632},
  doi = {10.1007/BF02280291},
  url = {https://doi.org/10.1007/BF02280291}
}

@inproceedings{Rat+23,
  author       = {Gaurav Rattan and
                  Tim Seppelt},
  title        = {Weisfeiler-Leman and Graph Spectra},
  booktitle    = {Proceedings of the 2023 {ACM-SIAM} Symposium on Discrete Algorithms ({SODA})},
  pages        = {2268--2285},
  publisher    = {{SIAM}},
  year         = {2023},
  doi           = {10.1137/1.9781611977554.ch87},
  url          = {https://doi.org/10.1137/1.9781611977554.ch87}
}

@InProceedings{Bab+79,
  author    = {Babai, L. and Ku{\v{c}}era, L.},
  title     = {Canonical Labelling of Graphs in Linear Average Time},
  booktitle = {Symposium on Foundations of Computer Science},
  year      = {1979},
  pages     = {39--46},
}

@article{Kie+13,
    title      = {On the Complexity of Equivalence and Minimisation for Q-multiplicity Automata},
    author     = {Stefan Kiefer and Andrzej Murawski and Joel Ouaknine and Bjoern Wachter and James Worrell},
    url        = {https://lmcs.episciences.org/908},
    doi        = {10.2168/LMCS-9(1:8)2013},
    journal    = {Logical Methods in Computer Science},
    issn       = {1860-5974},
    volume     = {Volume 9, Issue 1},
    eid        = 8,
    year       = {2013},
    month      = {Mar},
}

@article{Dvo+10,
author={Zden{\v{e}}k Dvo{\v{r}}{\'a}k},
title = {On recognizing graphs by numbers of homomorphisms},
journal = {Journal of Graph Theory},
volume = {64},
number = {4},
pages = {330-342},
doi = {10.1002/jgt.20461},
url = {https://onlinelibrary.wiley.com/doi/abs/10.1002/jgt.20461},
year = {2010}
}

@Book{Gro+17,
  title  = {Descriptive Complexity, Canonisation, and Definable Graph Structure Theory},
  year   = {2017},
  publisher = {Cambridge University Press},
  author = {Grohe, M.}
}

@InProceedings{Gro+21b,
  title = {Homomorphism Tensors and Linear Equations},
  author = {Grohe, Martin and Rattan, Gaurav and Seppelt, Tim},  
  doi = {arXiv:2111.11313v3},
  url = {https://arxiv.org/abs/2111.11313v3},
  year={2022},
}

@InProceedings{Del+18,
  author    = {H. Dell and M. Grohe and G. Rattan},
  title     = {Lov{\'{a}}sz Meets {W}eisfeiler and {L}eman},
  booktitle = {International Colloquium on Automata, Languages, and Programming},
  year      = {2018},
  pages     = {40:1--40:14},
}

@InProceedings{Imm+90,
  author    = {Immerman, N. and Lander, E.},
  title     = {Describing Graphs: {A} First-Order Approach to Graph Canonization},
  booktitle = {Complexity Theory Retrospective: {I}n Honor of Juris Hartmanis on the Occasion of His Sixtieth Birthday, July 5, 1988},
  year      = {1990},
  pages     = {59--81},
}

@article{Cai+92,
  author    = {Jin-Yi Cai and Martin F{\"u}rer and Neil Immerman},
  title     = {An optimal lower bound on the number of variables for graph identification},
  journal   = {Combinatorica},
  year      = {1992},
  volume    = {12},
  number    = {4},
  pages     = {389--410},
  doi       = {10.1007/BF01305232},
  url       = {https://doi.org/10.1007/BF01305232},
  issn      = {1439-6912},
}

@article{Bod+98,
author = {Bodlaender, Hans L.},
title = {A partial k-arboretum of graphs with bounded treewidth},
year = {1998},
issue_date = {Dec. 6, 1998},
publisher = {Elsevier Science Publishers Ltd.},
address = {GBR},
volume = {209},
number = {1–2},
issn = {0304-3975},
url = {https://doi.org/10.1016/S0304-3975(97)00228-4},
doi = {10.1016/S0304-3975(97)00228-4},
journal = {Theor. Comput. Sci.},
month = dec,
pages = {1–45},
numpages = {45},
}

@InProceedings{Flu+24,
  author =	{Fluck, Eva and Seppelt, Tim and Spitzer, Gian Luca},
  title =	{{Going Deep and Going Wide: Counting Logic and Homomorphism Indistinguishability over Graphs of Bounded Treedepth and Treewidth}},
  booktitle =	{32nd EACSL Annual Conference on Computer Science Logic (CSL 2024)},
  pages =	{27:1--27:17},
  series =	{Leibniz International Proceedings in Informatics (LIPIcs)},
  ISBN =	{978-3-95977-310-2},
  ISSN =	{1868-8969},
  year =	{2024},
  volume =	{288},
  editor =	{Murano, Aniello and Silva, Alexandra},
  publisher =	{Schloss Dagstuhl -- Leibniz-Zentrum f{\"u}r Informatik},
  address =	{Dagstuhl, Germany},
  URL =		{https://drops.dagstuhl.de/entities/document/10.4230/LIPIcs.CSL.2024.27},
  URN =		{urn:nbn:de:0030-drops-196704},
  doi =		{10.4230/LIPIcs.CSL.2024.27}
}

@article{Abr+21,
    author = {Abramsky, Samson and Shah, Nihil},
    title = {Relating structure and power: Comonadic semantics for computational resources},
    journal = {Journal of Logic and Computation},
    volume = {31},
    number = {6},
    pages = {1390-1428},
    year = {2021},
    month = {08},
    issn = {0955-792X},
    doi = {10.1093/logcom/exab048},
    url = {https://doi.org/10.1093/logcom/exab048},
}

@INPROCEEDINGS {Abr+17,
author = {S. Abramsky and A. Dawar and P. Wang},
booktitle = {2017 32nd Annual ACM/IEEE Symposium on Logic in Computer Science (LICS)},
title = {The pebbling comonad in Finite Model Theory},
year = {2017},
volume = {},
issn = {},
pages = {1-12},
doi = {10.1109/LICS.2017.8005129},
url = {https://doi.ieeecomputersociety.org/10.1109/LICS.2017.8005129},
publisher = {IEEE Computer Society},
address = {Los Alamitos, CA, USA},
month = {jun}
}

@misc{Cer+25,
      title={Caterpillar GNN: Replacing Message Passing with Efficient Aggregation},
      author={Marek \v{C}ern\'{y}},
      year={2025},
      eprint={2506.06784},
      archivePrefix={arXiv},
      primaryClass={cs.LG},
      url={https://arxiv.org/abs/2506.06784}, 
}

@book{Cha+90,
  title={Model theory},
  author={Chang, Chen Chung and Keisler, H Jerome},
  volume={73},
  year={1990},
  publisher={Elsevier}
}

@article{Bod+95,
author = {Bodlaender, Hans L. and Gilbert, John R. and Hafsteinsson, Hj\'{a}lmt\'{y}r and Kloks, Ton},
title = {Approximating treewidth, pathwidth, frontsize, and shortest elimination tree},
year = {1995},
issue_date = {March 1995},
publisher = {Academic Press, Inc.},
address = {USA},
volume = {18},
number = {2},
issn = {0196-6774},
url = {https://doi.org/10.1006/jagm.1995.1009},
doi = {10.1006/jagm.1995.1009},
journal = {J. Algorithms},
month = mar,
pages = {238-255},
numpages = {18}
}

@InProceedings{Sep+24,
  author =	{Seppelt, Tim},
  title =	{{An Algorithmic Meta Theorem for Homomorphism Indistinguishability}},
  booktitle =	{49th International Symposium on Mathematical Foundations of Computer Science (MFCS 2024)},
  pages =	{82:1--82:19},
  series =	{Leibniz International Proceedings in Informatics (LIPIcs)},
  ISBN =	{978-3-95977-335-5},
  ISSN =	{1868-8969},
  year =	{2024},
  volume =	{306},
  publisher =	{Schloss Dagstuhl -- Leibniz-Zentrum f{\"u}r Informatik},
  address =	{Dagstuhl, Germany},
  URL =		{https://drops.dagstuhl.de/entities/document/10.4230/LIPIcs.MFCS.2024.82},
  URN =		{urn:nbn:de:0030-drops-206387},
  doi =		{10.4230/LIPIcs.MFCS.2024.82}
}

@book{Klo+94,
  author    = {Ton Kloks},
  title     = {Treewidth: Computations and Approximations},
  series    = {Lecture Notes in Computer Science},
  volume    = {842},
  publisher = {Springer},
  address   = {Berlin, Heidelberg},
  year      = {1994},
  doi       = {10.1007/BFb0045375},
}

@article{Bal+15,
  title={A Canonical Form for Weighted Automata and Applications to Approximate Minimization},
  author={Borja Balle and P. Panangaden and Doina Precup},
  journal={2015 30th Annual ACM/IEEE Symposium on Logic in Computer Science},
  year={2015},
  pages={701-712},
}

@article{Tze+96,
author = {Tzeng, Wen-Guey},
title = {On path equivalence of nondeterministic finite automata},
year = {1996},
issue_date = {April 8, 1996},
publisher = {Elsevier North-Holland, Inc.},
address = {USA},
volume = {58},
number = {1},
issn = {0020-0190},
url = {https://doi.org/10.1016/0020-0190(96)00039-7},
doi = {10.1016/0020-0190(96)00039-7},
journal = {Inf. Process. Lett.},
month = apr,
pages = {43–46},
numpages = {4}
}

@article{Mon+24,
    title      = {The Pebble-Relation Comonad in Finite Model Theory},
    author     = {Yoàv Montacute and Nihil Shah},
    url        = {https://lmcs.episciences.org/10884},
    doi        = {10.46298/lmcs-20(2:9)2024},
    journal    = {Logical Methods in Computer Science},
    issn       = {1860-5974},
    volume     = {Volume 20, Issue 2},
    eid        = 9,
    year       = {2024},
    month      = {May},
}

@InProceedings{Mor+19,
  author    = {Christopher Morris and Martin Ritzert and Matthias Fey and William L. Hamilton and Jan Eric Lenssen and Gaurav Rattan and Martin Grohe},
  title     = {Weisfeiler and {L}eman Go Neural: Higher-Order Graph Neural Networks},
  booktitle = {{AAAI} Conference on Artificial Intelligence},
  year      = {2019},
  pages     = {4602--4609},
}

@article{Abr+23,
    title      = {Arboreal Categories: An Axiomatic Theory of Resources},
    author     = {Samson Abramsky and Luca Reggio},
    url        = {https://lmcs.episciences.org/9839},
    doi        = {10.46298/lmcs-19(3:14)2023},
    journal    = {Logical Methods in Computer Science},
    issn       = {1860-5974},
    volume     = {Volume 19, Issue 3},
    eid        = 14,
    year       = {2023},
    month      = {Aug},
}

@article{Ton+23,
title={Walking Out of the Weisfeiler Leman Hierarchy: Graph Learning Beyond Message Passing},
author={Jan T{\"o}nshoff and Martin Ritzert and Hinrikus Wolf and Martin Grohe},
journal={Transactions on Machine Learning Research},
issn={2835-8856},
year={2023},
url={https://openreview.net/forum?id=vgXnEyeWVY},
note={}
}

@inproceedings{Kri+22,
 author = {Kriege, Nils M.},
 booktitle = {Advances in Neural Information Processing Systems},
 pages = {20119--20132},
 publisher = {Curran Associates, Inc.},
 title = {Weisfeiler and Leman Go Walking: Random Walk Kernels Revisited},
 volume = {35},
 year = {2022}
}

@inproceedings{Xu+18,
title={How Powerful are Graph Neural Networks?},
author={Keyulu Xu and Weihua Hu and Jure Leskovec and Stefanie Jegelka},
booktitle={International Conference on Learning Representations},
year={2019},
url={https://openreview.net/forum?id=ryGs6iA5Km},
}

@article{Wei+68,
  author = {Weisfeiler, Boris and Lehman, A. A.},
  journal = {Nauchno-Technicheskaya Informatsia},
  number = {N9},
  pages = {12--16},
  title = {{A Reduction of a Graph to a Canonical Form and an Algebra Arising During This Reduction}},
  volume = {Ser. 2},
  year = 1968
}

@book{Wei+76,
  author    = {Boris Weisfeiler},
  title     = {On Construction and Identification of Graphs},
  series    = {Lecture Notes in Mathematics},
  volume    = {558},
  publisher = {Springer Berlin Heidelberg},
  year      = {1976},
  doi       = {10.1007/BFb0089374},
}

@article{Cor+07,
  title	= {Lp Distance and Equivalence of Probabilistic Automata},
  author	= {Corinna Cortes and Mehryar Mohri and Ashish Rastogi},
  year	= {2007},
  URL	= {http://www.cs.nyu.edu/~mohri/postscript/dist.pdf},
  journal	= {International Journal of Foundations of Computer Science},
  pages	= {761-780},
  volume	= {18}
}

@book{Sto+11,
  title={Pasteur's Quadrant: Basic Science and Technological Innovation},
  author={Stokes, D.E.},
  isbn={9780815719076},
  year={2011},
  publisher={Rowman \& Littlefield Publishers}
}

@book{Dur+98, 
place={Cambridge}, 
title={Biological Sequence Analysis: Probabilistic Models of Proteins and Nucleic Acids},
publisher={Cambridge University Press}, 
author={Durbin, Richard and Eddy, Sean R. and Krogh, Anders and Mitchison, Graeme}, 
year={1998}
}

@INPROCEEDINGS{Man+20,
  author={Man\v{c}inska, Laura and Roberson, David E.},
  booktitle={2020 IEEE 61st Annual Symposium on Foundations of Computer Science (FOCS)}, 
  title={Quantum isomorphism is equivalent to equality of homomorphism counts from planar graphs}, 
  year={2020},
  volume={},
  number={},
  pages={661-672},
  doi={10.1109/FOCS46700.2020.00067}
}

@InProceedings{Rob+23,
  author =	{Roberson, David E. and Seppelt, Tim},
  title =	{{Lasserre Hierarchy for Graph Isomorphism and Homomorphism Indistinguishability}},
  pages =	{101:1--101:18},
  series =	{Leibniz International Proceedings in Informatics (LIPIcs)},
  ISBN =	{978-3-95977-278-5},
  ISSN =	{1868-8969},
  year =	{2023},
  volume =	{261},
  address =	{Dagstuhl, Germany},
  URN =		{urn:nbn:de:0030-drops-181531},
  doi =		{10.4230/LIPIcs.ICALP.2023.101},
}

@article{Nig+24,
    author = "Kar, Prem Nigam and Roberson, David E. and Seppelt, Tim and Zeman, Peter",
    title = "{NPA Hierarchy for Quantum Isomorphism and Homomorphism Indistinguishability}",
    eprint = "2407.10635",
    archivePrefix = "arXiv",
    primaryClass = "quant-ph",
    doi = "10.4230/LIPIcs.ICALP.2025.105",
    journal = "Leibniz Int. Proc. Inf.",
    volume = "334",
    pages = "105:1--105:19",
    year = "2025"
}

@InProceedings{Kie+17,
author="Kiefer, Stefan
and Marusic, Ines
and Worrell, James",
editor="Pitts, Andrew",
title="Minimisation of Multiplicity Tree Automata",
booktitle="Foundations of Software Science and Computation Structures",
year="2015",
publisher="Springer Berlin Heidelberg",
address="Berlin, Heidelberg",
pages="297--311",
isbn="978-3-662-46678-0"
}

@InProceedings{Kie+14,
author="Kiefer, Stefan
and Wachter, Bj{\"o}rn",
editor="Esparza, Javier
and Fraigniaud, Pierre
and Husfeldt, Thore
and Koutsoupias, Elias",
title="Stability and Complexity of Minimising Probabilistic Automata",
booktitle="Automata, Languages, and Programming",
year="2014",
publisher="Springer Berlin Heidelberg",
address="Berlin, Heidelberg",
pages="268--279",
isbn="978-3-662-43951-7"
}

@article{Fli+74,
  author  = {Michel Fliess},
  title   = {Matrices de Hankel},
  journal = {Journal de Mathématiques Pures et Appliquées},
  volume  = {53},
  pages   = {197--222},
  year    = {1974}
}

@INPROCEEDINGS{Gro+23,
  author={Grohe, Martin and Lichter, Moritz and Neuen, Daniel},
  booktitle={2023 38th Annual ACM/IEEE Symposium on Logic in Computer Science (LICS)}, 
  title={The Iteration Number of the Weisfeiler-Leman Algorithm}, 
  year={2023},
  volume={},
  number={},
  pages={1-13},
  doi={10.1109/LICS56636.2023.10175741}
}

@inproceedings{Lic+19,
author = {Lichter, Moritz and Ponomarenko, Ilia and Schweitzer, Pascal},
title = {Walk refinement, walk logic, and the iteration number of the Weisfeiler-Leman algorithm},
publisher = {IEEE Press},
booktitle = {Proceedings of the 34th Annual ACM/IEEE Symposium on Logic in Computer Science},
articleno = {33},
numpages = {13},
location = {Vancouver, Canada},
series = {LICS '19}
}

\appendix

\section{Proofs of Propositions~\ref{prop:simpl_walks_and_walk_refiment} and~\ref{prop:swa_equivalence}}\label{app:auto_proofs}

\repeatproposition{simpl_walks_and_walk_refiment}
\begin{proof}
    Fix $t\ge 0$ and consider element $\vecw$ in $\Sw{(k,h)}{t}(G).$ 
    By definition, coloring $\chi_{k,h}$ always contains $\atp_\ell$, indicating whenever underlying tuples of $\vecw$ contained a repeated element. Denote this condition by (*).
    If (*) holds then by definition of $\xi^\tup{t}_\pm$ is $\vecw$ a $h$-colored $k$-simplicial walk in $G$ and its multiplicity is exactly the number of its occurrences in $G.$
    More closely, $\xi^\tup{t}_-$ collects tuples of outgoing vertices, and $\xi^\tup{t}_+$ collects tuples of incoming vertices.

    For the other direction, if (*) does not hold, consider incoming vertex $v$ causing duplication in $\vecu v$,
    with corresponding color $\vecw(i) = \chi_{(k,h)}(\vecu v)$ for some $i > 1$ (and \emph{possibly} $v$ is outgoing with color $\vecw(j)$ for some $j>i$). Note that $\vecw(i)$ is determined by $\vecw(i-1)$ (and possibly $\vecw(j)$ is determined by $\vecw(j+1)$),
    hence $\vecw'$ obtained from $\vecw$ by deleting $\vecw(i)$ (and possibly deleting $\vecw(j)$) and deleting vertex $v$ from tuples from $i$ to $t$ (or possibly to $j-1$), is determined by element $\vecw''$ in $\Sw{(k,h)}{t-1}(G)$, (or possibly in $\Sw{(k,h)}{t-2}(G)$), with one less duplication.
    Repeating the above step, we finitely obtain determination by $\hat \vecw$ with no duplications.
\qed\end{proof}

\repeatproposition{swa_equivalence}
\begin{proof}
    Denote the alphabet of $\Ac(G, k, h)$ by  $\Sig = \Sig(G, k, h)$.
    Similarly to (*) in the proof of \cref{prop:simpl_walks_and_walk_refiment}, 
    it suffices focus on color words $\vecc$ in $\Sw {(k,h)} t (G)$ whose colors do not originate from tuples of duplicated vertices.
    Here, we call color \emph{large} if it is a color of $(k+1)$-tuple, corresponding to the case $\chi_{(k,h)} = \atp_{k+1}$ of \cref{def:color_of_simplicial_walk}.
    
    For the backward implication, consider a word $\vecc = c_1 c_2 \dots c_t \In \Sw {(k,h)} {t-1} (G).$
    We split $\vecc$ into maximal parts $c_i\cdots c_j$ that does not contain a large color.
    For each $c_{i}\cdots c_{j}$ of those parts we use the corresponding letter $\Sig^{<\infty}.$
    Since $j-i \neq 0$, consider case $j-i = 1$ and use $(c_1, =, 1, c_1)\In\Sig$ corresponding to multiplicity of $(c_1).$
    Otherwise, $j-i > 1$, then for every $c_ic_\ip$ use  $(c_i, s, c_\ip)\In\Sig$, such that $s \In \{-, +\}$ depends on whenever $c_\ip$ was obtained by incoming or outgoing vertex, ($\xi_-^\tup\ip$ or $\xi_+^\tup\ip$), respectively, as determined from original tuple sizes.
    For each large color $c_{j} c_{j+1} c_{j+2}$, we notice that there appears necessarily incoming vertex 
    $u$ for original tuples of $c_{j}c_{j+1}$ and incoming vertex $v$ for $c_{j+1}c_{j+2}.$
    Since $u$ always incomes as the last one in the tuple we take some $p$ corresponding to the outcome of $v$,
    and use letter $(c_j, \tau, p, c_{j+2})$ where $\tau = \atp_2(u, v)$, so that $\atp_k$ contained in $c_j$ and $c_{j+2}$ together with $p$ and $\tau$ determine the large color in $\im(\atp_{k+1}).$
    
    For the forward implication, firstly, notice that $\sem{\Ac(G, k, h)}$ in addition to $\Sww {(k,h)}(G)$ contains words  $\vecw$ that are not starting with a color of $1$-tuple.
    This does not provide more information as multiplicity of $\vecw$ is the sum of multiplicities of all colored simplicial walks $c_1 c_2 \dots c_\ell \vecw.$
    Secondly, while matrix $D^-$ corresponds to $\xi_-^\tup\ip,$ the matrix $D^+$ symmetrically to $D^-$ but in contrast to $\xi_+^\tup\ip,$ extends each colored walk occurrence by placing incoming vertex $u$ on all positions of the underlying $\ell$-tuple instead of the last one only.
    Given a concrete WL color $c$ in $\im(\chi_{k,h})$, we notice that the count of $\chi_{k,h}(u_1, \dots, u_\ell, v)$ on $Q_k$ as well as the counts of
    $\chi_{k,h}(v, u_1, \dots, u_\ell), \dots, \chi_{k,h}(u_1, \dots, u_{\ell-1}, v, u_\ell)$ are mutually determined.
    Thus, this may only change the multiplicities by a constant factor dependent on $c.$
    Such a property of partition on $V(G)^k$ induced by (WL) coloring is also referred to as being shufflable \cite[Observation 3.1]{Gro+23}.
\qed\end{proof}

\section{Graph decompositions and a proof of Proposition~\ref{prop:decomposing_simplicial_walk}}\label{app:decompositions}

\subsubsection{Decompositions.}
For a given graph $F$, the pair of $(T, \beta)$,
where $T = (V(T), E(T))$ is a tree and $\beta\colon V(T) \to \powset{V(T)}$
is a mapping to \emph{bags} of vertices,
is a~\emph{tree decomposition} of $G$ if it holds:
\begin{enumerate}[label=(\roman*)]
    \item[(T1)] Every $uv \In E(F)$ is in at least one bag, that is,
    there exists $t \In V(T)$ such that $u, v \In \beta(t).$
    \item[(T2)] For each $v \In V(F)$, the induced subgraph $T\left[\{t \mid v \In \beta(t)\}\right]$ is connected and non-empty.
\end{enumerate}
The \emph{width} of a tree decomposition $(T, \beta)$ is defined as decreased size of the largest bag,
$\max_{t \in V(T)} |\beta(t)|-1.$
The \emph{treewidth} of a graph $F$, denoted by $\tw F$, is the minimum $k\In \Nat$ such that 
there exists a tree decomposition of $F$ of width at most $k.$

Moreover, the \emph{pathwidth} of a graph $F$, denoted by $\pw F$, is the minimum $k\In \Nat$ such that there exists a tree decomposition $(P, \beta)$ of $F$ of width at most $k$ such that $P$ is a path, that is, a tree of maximum degree two.
In this case, $(P, \beta)$ is a \emph{path decomposition} of~$F.$
Denote the class of graphs with pathwidth at most $k$ by $\mathcal{P}_{k,1}.$
We call $(P, \beta)$ a \emph{nice} path decomposition if the following holds:
\begin{enumerate}[label=(\roman*)]
    \item[(N1)] For every $st \In E(P)$ it holds $\beta(s) \sqcup \{u\} = \beta(t)$ or $\beta(s) = \beta(t) \sqcup \{u\}$. \label{item:nice1}
    \item[(N2)] There is $s \In V(P)$ of degree one such that $|\beta(s)| = 1.$\label{item:nice2}
\end{enumerate}
Note that our definition of nice path decomposition is a special case of nice tree decomposition,
also used in \cite{Bod+98, Flu+24}. Analogically to {\cite[Lemma~13.1.2]{Klo+94}}, for nice tree decomposition,
every graph $F$ has a nice path decomposition of width at most $\pw F.$

\repeatproposition{decomposing_simplicial_walk}
\begin{proof}
    Let $(P, \beta)$ be a nice path decomposition of $F$ of width at most $k$, assume that $V(P) = [t]$ and $E(P) = \{(i, i+1) \mid i \in [t-1]\}$ for some $t \In \Nat.$
    Take a $k$-skeleton $K(k)$ of a simplicial complex $K = \powset{V(F)}\setminus \{\emptyset\},$
    and define $\vecsig = (\beta(i) \mid i \in [t]).$
    Since $\vecsig$ is a walk in the Hasse diagram $\hasse {K(k)}$ by (N1) starting in the simplex $\vecsig(1)\In K(k)$ by (N2).
    Finally, if there is a vertex incoming more than once, or less than once, we get a contradiction with the definition path decomposition of $F$ by (T2).

    For the other direction, consider $k$-simplicial walk $\vecsig,$ of length $t$ decomposing in $F.$
    Use path $P$ on vertices $[t]$ as above, and define $\beta(i) = \vecsig(i)$ for each $i \In [t].$
    It remains to show that $(P, \beta)$ is a path decomposition of $F$ of width at most $k.$
    For $(P, \beta)$ condition (T1) is exactly (D2).
    Condition (D1) implies that each vertex $u \in V(F)$ is contained in bags of a single interval $\{i, i+1, \dots, j\} \subseteq [t]$ for some $i \le j$, and this is for path decomposition equivalent to (T2).
    Moreover, note this $(P, \beta)$ is also nice as (N1) and (N2) follow from the definition of a simplicial walk.
\qed\end{proof}

\section{Pebble Forest Covers and a proof of Proposition~\ref{prop:cd_path_decomposition}}\label{app:pebble_forest_covers}

\subsubsection{Treedepth.}
We use the following definition extended to the empty graph.
Let $F=(V, E)$ be a graph and $F_1, F_2, \dots F_p$ be its connected components for some $p \In \Nat$
 then a \emph{tree depth} of $F$ is defined recursively as follows:
\begin{align*}
    \td F = \begin{cases}
        0 & \text{ if } p = |V| = 0,\\
        1 + \max_{v \in V}(F[V\setminus \{v\}]) & \text{ if } p = 1,\\
        \max_{i \In [p]} \td {F_i} & \text{ otherwise .}
    \end{cases}
\end{align*}

\subsubsection{Pebble forest covers.}
We restate the definition as given in \cite{Abr+21, Rat+23}.
A chain in a poset $(P, \le)$ is a subset of pairwise comparable elements.
A \emph{forest} is a poset $(P, \le)$
where the subset $P_x = \{y \mid y \le x\}$ is a chain for each $x \In P.$
A \emph{depth of a forest} is the maximum length of a chain in $P.$
For a given graph $F$, 
a \emph{forest cover} is a forest $(V(G), \le)$ such that
for each $uv \In E(F)$ it holds that $u \le v$ or $v \le u.$
For such graph and its forest cover and integer $k\ge 0$,
a \emph{$k$-pebbling function} is a map $p\colon V(F) \to [k]$
such that for each $uv\In E(F)$ with $u\le v$, we have $p(u) \neq p(w)$ 
for any $w$ satisfying $u<w \le v.$
Finally, a \emph{$k$-pebble forest cover of depth $d$}
is a forest cover of depth $d$ together with a $k$-pebbling function.
We denote the graphs admitting a $k$-pebble forest cover of depth $d$ by $\mathcal{T}_k^d.$

\subsubsection{Labeled pebble forest covers.}
For a given $k$-labeled graph $F^\vecu$ and integer $\ell \ge k$, 
a $\ell$-pebble forest cover of $F^\vecu$ is a forest cover $(V(F), \le)$
and a $k$-pebbling function $p\colon V(F) \to [\ell]$,
such that the following holds:
\begin{enumerate}[label=(\roman*)]
    \item[(P1)] the set of labeled vertices $\sig = \{\vecu(1), \vecu(2), \ldots, \vecu(k)\}$ is a chain,
    \item[(P2)] every vertex $v \In V(F) \setminus \sig$ is less then all in $\sig$, that is $v \le \min \sig$,
    \item[(P3)] pebbling function on $\sig$ is injective, that is $p|_\sig$ is injective.
\end{enumerate}
Denote the class of all $\ell$-labeled graphs such that $\ell \le k$ with a \emph{$k$-pebble forest cover of depth $d$} by $\mathcal{LT}_{k}^d.$

\repeatproposition{cd_path_decomposition}
\begin{proof}
    Consider $F$ and its caterpillar decomposition $\vec D = ((\beta_i, L_i) \mid i \In [t])$ of width at most $k$ and height at most $h.$
    It holds that $\ell$-labeled graph in $L^\vecu \In \mathcal{LT}_{k}^d$ is of treedepth at most $\ell+d$ \cite{Rat+23}.
    From height of $\vec D$, we have for $L_i^{\stup \vecbeta i}$ that treedepth $\ell + d_i = |\beta_i|  + d_i \le k + h$ for $i \In [t].$
    By \cite[Lemma 11]{Bod+95}, stating $\pw{F} \le \td{F} - 1$ and definition of treedepth, we obtain $\pw{F}\le k + h - 1.$

    On the other hand, let $F$ be a graph and $(P, \beta)$ its path decomposition of width at most $k$.
    By \cref{prop:decomposing_simplicial_walk}, $(P, \beta)$ implies existence of a $k$-simplicial walk $\vecbeta$ of length $t'$ that is decomposing in $F$.
    Choose $L_i = F[\vecbeta(i)]$ for each $i \In [t']$ so that Conditions (C1) and (C3) are satisfied trivially and spine $S = F$ satisfying Condition (C2).
    We obtain a caterpillar decomposition $\vec D = ((L_i, \vecbeta(i)) \mid i \In [t'])$ of width at most $k$ and height at most $|\vecbeta(i)| + 0 - k \le k+1 -k \le  1.$
\qed\end{proof}

\end{document}